\documentclass[11p,reqno]{amsart}
\usepackage[colorlinks=true, pdfstartview=FitV, linkcolor=blue, citecolor=blue, urlcolor=blue]{hyperref}
\usepackage[usenames,dvipsnames]{xcolor}
\usepackage{latexsym}
\usepackage{amssymb}
\usepackage{amscd}
\usepackage{bm}
\usepackage{amsmath,amsthm}
\usepackage{setspace}
\usepackage{amssymb,amsthm,amsmath,mathrsfs,bm}
\usepackage{graphicx}
\usepackage{tabularx}
\usepackage{subfigure}
\usepackage{appendix}
\usepackage{enumerate}
\usepackage{cite}

\newcommand{\be}{\begin{eqnarray}}
\newcommand{\ee}{\end{eqnarray}}
\newcommand{\bez}{\begin{eqnarray*}}
\newcommand{\eez}{\end{eqnarray*}}

\numberwithin{equation}{section}

\numberwithin{equation}{section}
\newtheorem{theorem}{Theorem}[section]
\newtheorem{lemma}[theorem]{Lemma}
\newtheorem{coro}[theorem]{Corollary}
\newtheorem{prop}[theorem]{Proposition}

\theoremstyle{definition}
\newtheorem{define}[theorem]{Definition}
\newtheorem{remark}[theorem]{Remark}

\newcommand{\avg}[1]{\bigl\langle #1 \bigr\rangle}

\DeclareMathOperator{\Pf}{Pf}

\begin{document}

\title[Hermite--Pad\'{e} approximations with Pfaffian structures]{Hermite--Pad\'{e} approximations with Pfaffian structures: Novikov peakon equation and integrable lattices}

\date{}

\author{Xiang-Ke Chang}
\address{ LSEC, ICMSEC, Academy of Mathematics and Systems Science, Chinese Academy of Sciences, P.O.Box 2719, Beijing 100190, PR China; and School of Mathematical Sciences, University of Chinese Academy of Sciences, Beijing 100049, PR China.}
\email{changxk@lsec.cc.ac.cn}

\begin{abstract}
Motivated by the Novikov equation and its peakon problem, we propose a new mixed type Hermite--Pad\'{e} approximation whose unique solution is a sequence of polynomials constructed with the help of Pfaffians. These polynomials belong to the family of recently proposed partial-skew-orthogonal polynomials. The relevance of partial-skew-orthogonal polynomials is especially visible in the approximation problem germane to the Novikov peakon problem so that we obtain explicit inverse formulae in terms of Pfaffians by reformulating the inverse spectral problem for the Novikov multipeakons. 
Furthermore, we investigate two Hermite--Pad\'{e} approximations for the related spectral problem of the discrete dual cubic string, and show that these approximation problems can also be solved in terms of partial-skew-orthogonal polynomials and nonsymmetric Cauchy biorthogonal polynomials. This formulation results in a new correspondence among several integrable lattices. 
\end{abstract}

\keywords{
Hermite--Pad\'e approximation; Biorthogonal polynomials; Novikov equation; Multipeakons\\
To appear in Advances in Mathematics}

\maketitle
\tableofcontents

\section{Introduction} 
\subsection{On Camassa--Holm type equations and peakons}
Our motivation comes from a special kind of weak solutions called peaked solitary wave solutions (simply called \textit{peakons}), which are admitted  by a family of integrable partial differential equations (PDEs). These PDEs are sometimes called Camassa--Holm type equations for the reason that the first such mathematical model was proposed by Camassa and Holm in 1993 \cite{camassa1993integrable}:
\begin{equation*}\label{eq:ch}
m_t+(um)_x+mu_x=0, \qquad m=u-u_{xx},
\end{equation*}
and subsequently called the Camassa--Holm (CH) equation. 

In \cite{camassa1993integrable}, the CH equation was derived as a model for the unidirectional propagation of the shallow water waves over a flat bottom 
by executing an asymptotic expansion of the Hamiltonian for Euler's equations. Like the famous KdV equation, it admits many integrable properties, such as Lax integrability, bi-Hamiltonian structure and infinitely many conservation laws etc. In contrast to the KdV equation, it can lead to a meaningful model of wave-breaking, which is pursued by many researchers (e.g.  \cite{constantin-escher,whitham1974linear}) to describe nonlinear phenomena in nature, for example, the breakdown of regularity. Moreover, the CH equation allows the existence of multipeakon solutions of the form
\begin{equation*}
u(x,t)=\sum_{j=1}^Nm_j(t)e^{-|x-x_j(t)|}, \label{eq:CHpeakon}
\end{equation*}
which has no analog in the KdV theory.

The peakon solutions 
seem to capture main attributes of solutions of the CH equation: the breakdown of regularity can be interpreted as collisions of peakons, and the nature of 
long time asymptotics can be loosely described as peakons becoming 
free particles in the asymptotic region \cite{beals2000multipeakons,McKean-breakdown}. The CH peakons are orbitally stable, in the sense that a solution that initially is close to a peakon solution is also close to some peakon solution at a later time \cite{constantin2000stability,constantin2002stability}. This is one of reasons that peakons are of great interest.

On the other hand, the dynamics of the CH peakons can be described by an ODE system, which can be explicitly solved by the inverse spectral method\cite{beals1999multipeakons,beals2000multipeakons,eckhardt2014isospectral}. The key is the observation of an equivalence between the CH peakon spectral problem and a certain Krein type string 
\begin{equation*}
-\phi{''}(x)=zm(x)\phi(x),\qquad  x\in [-1,1],
\end{equation*}
with Dirichlet boundary conditions, which allows one to solve the corresponding inverse problem by employing Stieltjes' theory of continued fractions. As a result, the peakon solutions of the CH equation can be expressed in closed form in terms of Hankel determinants or orthogonal polynomials. Besides, due to the connection between the string problem and the Jacobi matrix spectral problem, the isospectral flow of which gives the Toda lattice, one can show that the CH peakon ODE system can be regarded as a negative flow of the finite Toda lattice in a well defined sense  \cite{beals2001peakons,ragnisco1996peakons}. We are particularly interested in this aspect of peakons.

Besides the CH equation, the other two extensively studied integrable PDEs with peakon solutions are the Degasperis--Procesi (DP) equation 
\cite{degasperis-procesi,lundmark2005degasperis}
\begin{equation*} 
  m_t+m_xu+3mu_x=0,   \qquad m=u-u_{xx},
\end{equation*}
and the Novikov equation \cite{hone2008integrable,hone2009explicit,novikov2009generalisations}
  \begin{equation*}
    m_t+(m_xu+3mu_x)u=0,   \qquad m=u-u_{xx}.
  \end{equation*}
Despite the superficial similarity, the DP equation has a mathematical structure quite different from that of the CH equation. The spatial equation in the Lax pair of the DP equation is of third order (rather than second), and is equivalent to the cubic string
\begin{equation*}
-\phi{'''}(x)=zm(x)\phi(x),\qquad  x\in [-1,1],
\end{equation*}
with Dirichlet-like boundary conditions. Its forward spectral problem was analyzed in \cite{lundmark2005degasperis} and the inverse spectral problem was solved by exploiting certain Hermite--Pad\'{e} approximation problems, which lead to explicit solution formulae in terms of Cauchy bimoment determinants for the DP peakons. The study on DP peakons later motivated the concept of Cauchy biorthogonal polynomials (CBOPs) \cite{bertola2010cauchy}, an isospectral deformation  of which gives the C-Toda lattice \cite{chang2018degasperis}. We also remark that CBOPs can be useful in the investigation of various problems in random matrix theory \cite{bertola2009cauchy,bertola2014cauchy}. 

Regarding the Novikov equation, it has cubic nonlinearity instead of second. It admits a $3\times3$ Lax pair, the spatial part of which can be transformed into the dual cubic string in the finite interval.  Its forward and inverse spectral problems in the discrete case were originally studied in \cite{hone2009explicit}, where the dual cubic string in the finite interval was introduced, while the inverse spectral problem was directly solved in \cite{mohajer2012inverse} according to the analysis on the real axis without any additional transformations to the finite interval. In both \cite{hone2009explicit,mohajer2012inverse}, the inverse spectral problem was solved by introducing specific Hermite--Pad\'{e} approximation problems so that explicit solution formulae of the Novikov peakons were given by Cauchy bimoment determinants as well. 

Recently, the argument in \cite{chang2018application} showed that it is more efficient to formulate multipeakons of the Novikov equation in terms of \textit{Pfaffians} (see Appendix \ref{app_pf} for an introduction). The main motivation was that \textit{Pfaffians} provide a more natural algebraic tool than \textit{determinants} to describe the solutions for the BKP type equations \cite{hirota2004direct}---a larger class of flows associated with the transformation groups $O(\infty)$ that the Novikov equation belongs to. Consequently, in \cite{chang2018application}, the Pfaffian expressions for the Novikov peakon solutions were successfully formulated by rewriting the results obtained in \cite{hone2009explicit} and the corresponding peakon ODEs are interpreted as an isospectral flow on a manifold cut out by Pfaffian identities. This prompts the following question: \textbf{can one obtain the Pfaffian formulae for the Novikov multipeakons by reformulating the inverse problem so that the Paffian structure appears front and center? } To this end, two alternative ways are provided in Section \ref{sec:NVpeakon} and Section \ref{sec:dual_cubic} respectively, where the one in Section \ref{sec:NVpeakon} seems more compelling.

\subsection{On Hermite--Pad\'{e} approximations} The inverse spectral problems related to peakons lead to multiple connections with classical analysis such as: orthogonal polynomials, Pad\'e approximations and continued fractions. In the CH peakon case \cite{beals1999multipeakons,beals2000multipeakons}, the solution of the inverse problem relies on a theorem of Stieltjes on continued fractions \cite{stieltjes1894}, the mechanism of which can be interpreted as Pad\'e approximations to a Weyl function. As for the DP and Novikov peakon cases \cite{lundmark2005degasperis,hone2009explicit}, the Hermite--Pad\'e type rational approximations to a set of
several Weyl functions play important roles in the corresponding inverse spectral problems. It is known that ordinary continued fractions can be considered as products of sequences of $2\times2$ matrices. Since $3\times3$ matrices are involved in the DP and Novikov peakon cases, the corresponding products can be interpreted as generalizations of ordinary continued fractions. Correspondingly, Hermite--Pad\'e type approximations are simultaneous rational approximations of several analytic functions and they are considered as generalizations of the well-known Pad\'e approximation problems for one analytic function.

Hermite--Pad\'e type approximations originated in the study of simultaneous rational approximation of systems of exponentials by Hermite in 1873. Subsequently, they have been used in other  related problems of number theory. In recent years, Hermite--Pad\'e type approximations have received increasing attention because of their appearance in other areas such as the theories of multiple orthogonal polynomials, random matrix, non-intersecting Brownian motions \cite{bleher2004random,daems2007multiple,kuijlaars2010multiple,lago2021introduction} and as well as peakon problems \cite{hone2009explicit,lago2019mixed,lundmark2005degasperis,lundmark2016inverse}. For more recent developments on other related topics, one might refer to\cite{aptekarev2016discrete,aptekarev2017hermite,lago2015convergence} etc.

Traditionally, Hermite--Pad\'e approximations are classified into two types, called type I and type II. More exactly, given a family of analytic functions ${\bf{f}} = (f_1,\ldots,f_m)$ in some domain $D$ of the extended complex plane containing $\infty$, for fixed $\textbf{n}=(n_1,n_2,\ldots,n_m) \in \mathbb{Z}_+^m\backslash\{\textbf{0}\}$ with $|\textbf{n}|=n_1+n_2+\cdots+n_m$, a construction of type I Hermite--Pad\'e approximations means
\begin{define}[Type I]
Seek polynomials $a_{0},a_{1},a_{2},\ldots,a_{m}$ satisfying
\begin{enumerate}[(i)]
\item $\deg{a_{j}}\leq n_j-1$, $j=1, \ldots, m,$ not all identically equal to zero,
\item $a_{0}-\sum_{j=1}^ma_{j}f_j=\mathcal{O}(1/z^{|{\bf{n}}|}),\,z\rightarrow\infty.$
\end{enumerate}
\end{define}
\noindent On the other hand, a construction of type II Hermite--Pad\'e approximations means
\begin{define}[Type II]
Seek polynomials $Q_{0},P_{1},P_{2},\ldots,P_{m}$ satisfying
\begin{enumerate}[(i)]
\item $\deg{Q_{0}}\leq |{\bf{n}}|$,
\item $Q_{0}f_j-P_{j}=\mathcal{O}(1/z^{n_j+1}),\,j=1, \ldots, m,\,z\rightarrow\infty$.
\end{enumerate}
\end{define}
\noindent Very recently, there has been a significant increase in the number of new results concerning a combination of the two (called mixed type) Hermite--Pad\'e approximations. The approximations appearing in the DP, Novikov and Geng-Xue peakon problems \cite{hone2009explicit,lundmark2005degasperis,lundmark2016inverse} are nothing but Hermite--Pad\'e approximations of mixed type. Moreover, it is noted that approximations of Nikishin systems (see e.g. \cite{lago2015convergence,lago2019mixed,lago2021introduction} for more details on such systems of functions) are involved in these peakon problems.

Usually, the unique solution of a Hermite--Pad\'e approximation can be derived by associating the approximations with linear systems so that determinant expressions naturally appear by using Cramer's rule. As far as we know, there is only one interesting example that the generalised inverse vector Pad\'e
approximants involve Pfaffian representations in the solutions \cite{graves1996cayley,graves1997problems}.  Thus the following question presents itself:  \textbf{is there any other Hermite--Pad\'e approximation problem whose solution is connected to natural Pfaffian structures?} Motivated by the Novikov peakon problem, we propose one novel approximation problem with Pfaffian structures in Section \ref{sec:HP_pf}. 

In fact, we associate the new Hermite--Pad\'{e} approximation problem with a specific sequence of partial-skew-orthogonal polynomials (PSOPs) in \cite{chang2018partial}.
Recall that the concept of PSOPs, introduced in \cite{chang2018partial},  was motivated by a  random matrix model called the Bures random ensemble as well as the hints from the formulation of the Novikov peakon solution in terms of Pfaffians.  In this paper, it is shown that a specific sequence of PSOPs  can be naturally used to solve the peakon inverse problem arising in the Novikov equation.

%

\subsection{Contributions and outline}
In summary, our new contributions include:
\begin{enumerate}[(i)]
\item Motivated by the inverse spectral problem for Novikov multipeakons, a novel Hermite--Pad\'{e} approximation problem is proposed with polynomials expressed by Pfaffians as its unique solution. These polynomials belong to a class of the so-called PSOPs introduced in \cite{chang2018partial}.
\item The explicit formulae of the Novikov multipeakons are obtained from the analysis on the real axis of the inverse spectral problem. The inverse problem involves only one Hermite--Pad\'{e} approximation problem with Pfaffian structures, in contrast to previous works \cite{hone2009explicit,mohajer2012inverse}. The role of PSOPs is highlighted in the approximation problems germane to the inverse problem.
\item Internal structures of Hermite--Pad\'{e} approximation problems involved in the discrete dual cubic string related to Novikov multipeakons are investigated, as a result of which, intrinsic relations between a family of specific nonsymmetric CBOPs and PSOPs are obtained together with underlying relationships among the corresponding integrable lattices including the Novikov peakon lattice, B-Toda lattice \cite{chang2018application} and an integrable lattice proposed in \cite{chang2021two}.
\end{enumerate}

Our paper is organized as follows. In Section \ref{sec:det_pf}, some preliminary definitions and properties of certain bimoment determinants and Pfaffians are presented. In Section \ref{sec:HP_pf}, we propose a novel Hermite--Pad\'e approximation problem with natural appearance of Pfaffian structures, which is subsequently used in Section \ref{sec:NVpeakon} to solve an inverse spectral problem on the real axis in an alternative way so that the explicit formulae in terms of Pfaffians for the Novikov peakons are naturally obtained. In Section \ref{sec:dual_cubic}, we reconsider the dual cubic string, which is obtained by a Liouville transformation of the spectral problem for the Novikov equation, and show that the solutions to two related Hermite--Pad\'e approximation problems can be associated with PSOPs and nonsymmetric CBOPs. Finally, the underlying relationships between a family of specific nonsymmetric CBOPs and PSOPs together with the corresponding integrable lattices are investigated by studying products of $3\times 3$ transition matrices.

\section{On certain bimoment determinants and Pfaffians}\label{sec:det_pf}
Let  $d\mu$  be a positive Stieltjes measure on $\mathbb{R}_+$ with finite moments 
\begin{equation*}
\beta_j=\int x^jd\mu(x)
\end{equation*}
and finite bimoments with respect to the Cauchy kernel $\frac{1}{x + y}$
\begin{equation*}
I_{i,j}=\iint \frac{x^iy^j}{x+y}d\mu(x)d\mu(y).
\end{equation*}
Note that the 
bimoments based on a skew kernel $\frac{y-x}{x + y}$
\begin{equation*}
J_{i,j}=\iint \frac{y-x}{x+y}x^iy^jd\mu(x)d\mu(y)
\end{equation*}
are then finite as well. 

Obviously, we have the following Lemma.
\begin{lemma}As for the moments, the following identities hold
\begin{align}\label{rel:IJb}
I_{i+1,j}+I_{i,j+1}=\beta_i\beta_j,\qquad J_{i,j}=I_{i,j+1}-I_{i+1,j}=\beta_i\beta_j-2I_{i+1,j}=2I_{i,j+1}-\beta_i\beta_j.
\end{align}
\end{lemma}
\begin{define}\label{def:FG}
For $k\geq1$, let $F_k^{(i,j)}$ denote the determinant of the $k\times k$ bimoment matrix which starts with $I_{i,j}$ at the upper left corner:
\begin{align*}
F_k^{(i,j)}=\left|
\begin{array}{cccc}
I_{i,j}&I_{i,j+1}&\cdots & I_{i,j+k-1}\\
I_{i+1,j}&I_{i+1,j+1}&\cdots&I_{i+1,j+k-1}\\
\vdots&\vdots&\ddots&\vdots\\
I_{i+k-1,j}&I_{i+k-1,j+1}&\cdots&I_{i+k-1,j+k-1}
\end{array}
\right|=F_k^{(j,i)}.
\end{align*}
Let $F_0^{(i,j)}=1$ and $F_k^{(i,j)}=0$ for $k<0$.

For $k\geq2$, let $G_k^{(i,j)}$ denote the $k\times k$ determinant
\begin{align*}
G_k^{(i,j)}=\left|
\begin{array}{ccccc}
I_{i,j}&I_{i,j+1}&\cdots & I_{i,j+k-2}&\beta_{i}\\
I_{i+1,j}&I_{i+1,j+1}&\cdots&I_{i+1,j+k-2}&\beta_{i+1}\\
\vdots&\vdots&\ddots&\vdots&\vdots\\
I_{i+k-1,j}&I_{i+k-1,j+1}&\cdots&I_{i+k-1,j+k-2}&\beta_{i+k-1}
\end{array}
\right|.
\end{align*}
Let $G_1^{(i,j)}=\beta_i$ and $G_k^{(i,j)}=0$ for $k<1$.
\end{define}
The positivity of these determinants are guaranteed by the following results, whose proof can be found in e.g. \cite{bertola2010cauchy,lundmark2016inverse}. Actually, the proof is similar to that for Heine's formula on Hankel determinants.
\begin{prop}\label{prop:FG}
For $k\geq1$, the determinants $F_k^{(i,j)},G_k^{(i,j)}$ admit the integral representations 
\begin{align*}
&F_k^{(i,j)}=\displaystyle\iint_{\substack{0<x_1<\cdots<x_k\\0<y_1<\cdots<y_k}}\frac{\Delta_{[k]}(\mathbf{x})^2\Delta_{[k]}(\mathbf{y})^2}{\Gamma_{[k],[k]}(\mathbf{x},\mathbf{y})}\prod_{p=1}^k\prod_{q=1}^k(x_p)^i(y_q)^jd\mu(x_p)d\mu(y_q),\\
&G_k^{(i,j)}=\displaystyle\iint_{\substack{0<x_1<\cdots<x_k\\0<y_1<\cdots<y_{k-1}}}\frac{\Delta_{[k]}(\mathbf{x})^2\Delta_{[k-1]}(\mathbf{y})^2}{\Gamma_{[k],[k-1]}(\mathbf{x},\mathbf{y})}\prod_{p=1}^k\prod_{q=1}^{k-1}(x_p)^i(y_q)^jd\mu(x_p)d\mu(y_q),
\end{align*}
where $[k]$ denotes the index set $\{1,2,\ldots,k\}$, and, for two index sets $I, J$, 
\begin{align*}
\Delta_J(\mathbf{x})=\prod_{i<j\in J}(x_j-x_i), 
  \quad \Gamma_{I,J}(\mathbf{x};\mathbf{y})=\prod_{i\in I}\prod_{j\in J}(x_i+y_j),  
  \end{align*}
  along with the convention
\begin{align*}
&\Delta_\emptyset(\mathbf{x})=\Delta_{\{i\}}(\mathbf{x})=\Gamma_{\emptyset,J}(\mathbf{x};\mathbf{y})=\Gamma_{I,\emptyset}(\mathbf{x};\mathbf{y})=1.
\end{align*}
Furthermore, 
\begin{enumerate}[(i)]
\item if $\mu(x)$ has infinitely many points of increase, then $F_k^{(i,j)}$ and $G_k^{(i,j)}$ are always positive for any $k\geq1$.
\item if $\mu(x)$ has only finitely many points of increase, say $K$, or equivalently $d\mu(x)=\sum_{i=1}^Ka_i\delta(x-x_i)dx$ with positive $a_i$, then $F_k^{(i,j)}$ and $G_k^{(i,j)}$ are positive for $1\leq k\leq K$, while vanish for any $k>K$.
\end{enumerate}
\end{prop}

\begin{remark}
Due to the multiple integral representations, $F_k^{(i,j)}$ is closely related to the partition function of the Cauchy two-matrix model \cite{bertola2009cauchy,bertola2014cauchy}. 
\end{remark}

Now we introduce Pfaffians based on entries $J_{i,j}$ and $\beta_j$, as we see below, which are of different forms according to parity of the order.
\begin{define} \label{def:tau}
For $k\geq1$, let $\tau_k^{(l)}$ denote the Pfaffian\footnote{See Appendix \ref{app_pf} for introduction, notations and some useful formulae. It is noted that we employ the usual notation for Pfaffians in this paper, while the Hirota's notation (see Remarks \ref{rem:pf1}) has been adopted in our previous works \cite{chang2018partial,chang2018application}.}

\begin{align*}
\tau_{k}^{(l)}=
\left\{
\begin{array}{ll}
\Pf\left(\left(J_{l+i,l+j}\right)_{i,j=0}^{2m-1}\right),&k=2m,\\
\ \\
\Pf\left(\left(
\begin{array}{ccc}
0&\vline&
\begin{array}{c}
\beta_{l+j}
\end{array}\\
\hline
-\beta_{l+i}
&\vline&J_{l+i,l+j}
\end{array}
\right)_{i,j=0}^{2m-1}\right),&k=2m-1.
\end{array}
\right.
\end{align*}
Let $\tau_0^{(l)}=1$ and $\tau_k^{(l)}=0$ for $k<0$. For convenience, we usually write $\tau_k$ as an abbreviation of $\tau_k^{(0)}$.
\end{define}
 

The positivity properties of these Pfaffians follow immediately from the integral representation (\ref{integral_tau}) below, which was originally proved in \cite{chang2018application,chang2018partial} using de Bruijn's formula. Here we provide a new straightforward proof based only on Schur's Pfaffian identities  \eqref{id_schur} and \eqref{id_schur_odd}. 
\begin{prop}
For $k\geq1$, the Pfaffian $\tau_{k}^{(l)}$ admits the integral representation
\begin{equation}\label{integral_tau}
\tau_k^{(l)}=\displaystyle\idotsint\limits_{{0<x_1<\cdots<x_k}}\frac{\Delta_{[k]}(\mathbf{x})^2}{\Gamma_{[k]}(\mathbf{x})}\prod_{p=1}^k(x_p)^ld\mu(x_p), 
\end{equation}
where $[k]$ denotes the index set $\{1,2,\ldots,k\}$, and, for two index sets $I, J$, 
\begin{align*}
\Delta_J(\mathbf{x})=\prod_{i<j\in J}(x_j-x_i), \quad  \Gamma_{J}(\mathbf{x})=\prod_{i<j\in J}(x_j+x_i), 
  \end{align*}
  along with the convention
\begin{align*}
&\Delta_\emptyset(\mathbf{x})=\Delta_{\{i\}}(\mathbf{x})=\Gamma_{\emptyset}(\mathbf{x})=\Gamma_{\{i\}}(\mathbf{x})=1.
\end{align*}
Furthermore, 
\begin{enumerate}[(i)]
\item if $\mu(x)$ has infinitely many points of increase, then $\tau_{k}^{(l)}$ are always positive for any $k\geq1$.
\item if $\mu(x)$ has only finitely many points of increase, say $K$, or equivalently $d\mu(x)=\sum_{i=1}^Ka_i\delta(x-x_i)dx$ with positive $a_i$, then $\tau_{k}^{(l)}$ are positive for $1\leq k\leq K$, while vanish for any $k>K$.
\end{enumerate}
\end{prop}
\begin{proof}
It is sufficient to confirm the multiple integral formula, which can be obtained by use of Schur's Pfaffian identities \eqref{id_schur}--\eqref{id_schur_odd}. Here we only give the detailed proof for the even case $k=2m$. Similar argument works for the odd case.  

First, we have from Definition \ref{def_pf} of the Pfaffian expansion and the integral expressions of the Pfaffian entries
\begin{align*}
\tau_{2m}^{(l)}&=\sum_{P}(-1)^PJ_{l+i_1-1,l+i_2-1}J_{l+i_3-1,l+i_4-1}\cdots J_{l+i_{2m-1}-1,l+i_{2m}-1}\\
&=\sum_{P}(-1)^P\idotsint \frac{x_{i_2}-x_{i_1}}{x_{i_1}+x_{i_2}}\frac{x_{i_4}-x_{i_3}}{x_{i_3}+x_{i_4}}\cdots \frac{x_{i_{2m}}-x_{i_{2m-1}}}{x_{i_{2m-1}}+x_{i_{2m}}}x_1^lx_2^{l+1}\cdots x_{2m}^{l+2m-1}d\mu(x_1)\cdots d\mu(x_{2m}).
\end{align*}
Then, using \eqref{id_schur} and Definition \ref{def_pf} of the Pfaffian expansion, we arrive at 
\begin{align*}
\tau_{2m}^{(l)}&=\idotsint \frac{\Delta_{[2m]}(\mathbf{x})}{\Gamma_{[2m]}(\mathbf{x})}\prod_{p=1}^{2m}(x_p)^{l+p-1}d\mu(x_1)\cdots d\mu(x_{2m}).
\end{align*}
According to symmetrization over the integral variables, we finally obtain
\begin{align*}
\tau_{2m}^{(l)}&=\sum_{\sigma\in S_{2m}}\quad\idotsint\limits_{{0<x_1<\cdots<x_{2m}}} \frac{\Delta_{[2m]}(\mathbf{x}_\sigma)}{\Gamma_{[2m]}(\mathbf{x}_\sigma)}\prod_{p=1}^{2m}(x_p)^{\sigma_p-1}\prod_{p=1}^{2m}(x_p)^ld\mu(x_1)\cdots d\mu(x_{2m})\\
&=\sum_{\sigma\in S_{2m}}\quad\idotsint\limits_{{0<x_1<\cdots<x_{2m}}} \frac{\Delta_{[2m]}(\mathbf{x})}{\Gamma_{[2m]}(\mathbf{x})}(-1)^\sigma\prod_{p=1}^{2m}(x_p)^{\sigma_p-1}\prod_{p=1}^{2m}(x_p)^ld\mu(x_p)\\
&=\displaystyle\idotsint\limits_{{0<x_1<\cdots<x_{2m}}}\frac{\Delta_{[2m]}(\mathbf{x})^2}{\Gamma_{[2m]}(\mathbf{x})}\prod_{p=1}^{2m}(x_p)^ld\mu(x_p),
\end{align*}
where $S_{2m}$ stands for the symmetric group on $2m$ elements and the last equality follows from the definition of the Vandermonde determinant.
\end{proof}
\begin{remark}
Here we provide a straightforward proof without using de Bruijn's formula as in \cite{chang2018application,chang2018partial} and note that no time deformation is involved so far. These Pfaffians are closely associated with the partition function for the Bures random matrix ensemble \cite{forrester2016relating}.
\end{remark}

\begin{remark}
Some of the quantities used  in this paper have  appeared in the previous works, however, presented in different notations. For example, $F_k^{(i,j)}$ agrees with the notation $\Delta_k^{(i,j)}$ in \cite{bertola2010cauchy}, and $D_k^{(i,j)}$ in \cite{hone2009explicit}.   $G_k^{(0,1)}$ and $G_k^{(0,2)}$ are equivalent to the quantities $D_k^{'}$  and $D_k^{''}$ in  \cite{hone2009explicit,lundmark2005degasperis}.  For the Pfaffian $\tau_k^{(l)}$, we have $\tau_k^{(-1)}=t_k,$ $\tau_k^{(0)}=u_k,$ $\tau_k^{(1)}=v_k$ in the notation of  \cite{hone2009explicit,lundmark2005degasperis}.
\end{remark}

There exist intimate relations among the determinants $F_k^{(i,j)},G_k^{(i,j)}$ and the Pfaffians $\tau_k^{(l)}$, some of which can be found in the references (e.g.\cite{bertola2009cauchy,forrester2016relating,hone2009explicit,kohlenberg2007inverse,lundmark2005degasperis}). Here, based on the identities \eqref{det_pf_odd}--\eqref{det_pf_even} connecting determinants and Pfaffians, we shall present a new and straightforward approach to derive some formulae, of which the one for $F_k^{(l+1,l-1)}$ seems new.

\begin{theorem}\label{prop:FGtau}
For any integer $k,l$, the following formulae hold: 
\begin{equation}\label{rel:FGtau}
\begin{aligned}
&F_k^{(l+1,l)}=F_k^{(l,l+1)}=\frac{(\tau_k^{(l)})^2}{2^k},\qquad G_k^{(l,l+1)}=\frac{\tau_k^{(l)}\tau_{k-1}^{(l)}}{2^{k-1}},  \qquad G_k^{(l,l+2)}=\frac{\tau_k^{(l)}\tau_{k-1}^{(l+1)}}{2^{k-1}},\\
& F_k^{(l,l)}=\frac{1}{2^k}\left(\tau_k^{(l)}\tau_k^{(l-1)}-\tau_{k-1}^{(l)}\tau_{k+1}^{(l-1)}\right),\quad F_k^{(l+1,l-1)}=\frac{1}{2^k}\left(\tau_k^{(l)}\tau_k^{(l-1)}+\tau_{k-1}^{(l)}\tau_{k+1}^{(l-1)}\right). 
\end{aligned}
\end{equation}
\end{theorem}
\begin{proof}
First, regarding the $k\times k$ determinant $F_k^{(l+1,l)}$, we can rewrite it as a $(k+1)\times(k+1)$ determinant as follows
\begin{align*}
 F_k^{(l+1,l)}=\frac{(-1)^k}{2^k}\left|
\begin{array}{ccccc}
1&0&0&\cdots&0\\
\beta_l&-2I_{l+1,l}&-2I_{l+1,l+1}&\cdots & -2I_{l+1,l+k-1}\\
\beta_{l+1}&-2I_{l+2,l}&-2I_{l+2,l+1}&\cdots & -2I_{l+2,l+k-1}\\
\vdots&\vdots&\vdots&\ddots&\vdots\\
\beta_{l+k-1}&-2I_{l+k,l}&-2I_{l+k,l+1}&\cdots & -2I_{l+k,l+k-1}
\end{array}
\right|,
\end{align*}
where the equality follows from expanding along the first row and taking out a factor of $-2$ from every column except the first. Then,
by adding the first column multiplied by $\beta_{l+j-2}$ from the $j$-th column with $j=2,3\cdots,k+1$, and using  the relation $J_{i,j}=\beta_i\beta_j-2I_{i+1,j}$ from \eqref{rel:IJb}, we get
\begin{align*}
 F_k^{(l+1,l)}=&\frac{(-1)^k}{2^k}\left|
\begin{array}{ccccc}
1&\beta_l&\beta_{l+1}&\cdots&\beta_{l+k-1}\\
\beta_l&J_{l,l}&J_{l,l+1}&\cdots & J_{l,l+k-1}\\
\beta_{l+1}&J_{l+1,l}&J_{l+1,l+1}&\cdots & J_{l+1,l+k-1}\\
\vdots&\vdots&\vdots&\ddots&\vdots\\
\beta_{l+k-1}&J_{l+k-1,l}&J_{l+k-1,l+1}&\cdots & J_{l+k-1,l+k-1}
\end{array}
\right|\\
=&\frac{(-1)^k}{2^k}\left|
\begin{array}{ccccc}
J_{l,l}&J_{l,l+1}&\cdots & J_{l,l+k-1}\\
J_{l+1,l}&J_{l+1,l+1}&\cdots & J_{l+1,l+k-1}\\
\vdots&\vdots&\ddots&\vdots\\
J_{l+k-1,l}&J_{l+k-1,l+1}&\cdots & J_{l+k-1,l+k-1}
\end{array}
\right|\\
&-\frac{(-1)^k}{2^{k}}\left|
\begin{array}{ccccc}
0&\beta_l&\beta_{l+1}&\cdots&\beta_{l+k-1}\\
-\beta_l&J_{l,l}&J_{l,l+1}&\cdots & J_{l,l+k-1}\\
-\beta_{l+1}&J_{l+1,l}&J_{l+1,l+1}&\cdots & J_{l+1,l+k-1}\\
\vdots&\vdots&\vdots&\ddots&\vdots\\
-\beta_{l+k-1}&J_{l+k-1,l}&J_{l+k-1,l+1}&\cdots & J_{l+k-1,l+k-1}
\end{array}
\right|,
\end{align*}
where the second equality is obtained by observing 
$$(1,\beta_l,\ldots,\beta_{l+k-1})^T=(1,0,\ldots,0)^T-(0,-\beta_l,\ldots,-\beta_{l+k-1})^T,$$
and using the linear property along the first column  of the $(k+1)\times (k+1)$ determinant  in the left hand side of the equality. Then the Pfaffian expression for $F_k^{(l+1,l)}$ follows from checking the  parity of $k$ and the fact that any skew-symmetry determinant of odd order vanishes.

As for $G_k^{(l,l+1)}$, using $J_{i,j}=2I_{i,j+1}-\beta_i\beta_j$ in \eqref{rel:IJb} and similar column operations, we have
\begin{align*}
G_k^{(l,l+1)}
&=\frac{1}{2^{k-1}}\left|
\begin{array}{ccccc}
2I_{l,l+1}&2I_{l,l+2}&\cdots & 2I_{l,l+k-1}&\beta_{l}\\
2I_{l+1,l+1}&2I_{l+1,l+2}&\cdots&2I_{l+1,l+k-1}&\beta_{l+1}\\
\vdots&\vdots&\ddots&\vdots&\vdots\\
2I_{l+k-1,l+1}&2I_{l+k-1,l+2}&\cdots&2I_{l+k-1,l+k-1}&\beta_{l+k-1}
\end{array}
\right|\\
&=\frac{1}{2^{k-1}}\left|
\begin{array}{ccccc}
J_{l,l}&J_{l,l+1}&\cdots & J_{l,l+k-2}&\beta_{l}\\
J_{l+1,l}&J_{l+1,l+1}&\cdots &J_{l+1,l+k-2}&\beta_{l+1}\\
\vdots&\vdots&\ddots&\vdots&\vdots\\
J_{l+k-1,l}&J_{l+k-1,l+1}&\cdots & J_{l+k-1,l+k-2}&\beta_{l+k-1}
\end{array}
\right|,
\end{align*}
from which, in conjunction with the identities \eqref{det_pf_odd}--\eqref{det_pf_even}, the conclusion for $G_k^{(l,l+1)}$ follows.  Similarly, the proof on $G_k^{(l,l+2)}$ can be completed by observing
\begin{align*}
G_k^{(l,l+2)}
&=\frac{1}{2^{k-1}}\left|
\begin{array}{ccccc}
J_{l,l+1}&J_{l,l+2}&\cdots & J_{l,l+k-1}&\beta_{l}\\
J_{l+1,l+1}&J_{l+1,l+2}&\cdots &J_{l+1,l+k-1}&\beta_{l+1}\\
\vdots&\vdots&\ddots&\vdots&\vdots\\
J_{l+k-1,l+1}&J_{l+k-1,l+2}&\cdots & J_{l+k-1,l+k-1}&\beta_{l+k-1}
\end{array}
\right|
\end{align*}
and employing the identities \eqref{det_pf_odd}--\eqref{det_pf_even}.

Finally, let's turn to $F_k^{(l,l)}$ and $F_k^{(l+1,l-1)}$. Following the procedure for $F_k^{(l+1,l)}$ above and using the relation $J_{i,j}=\beta_i\beta_j-2I_{i+1,j}$ from \eqref{rel:IJb}, it is not hard to arrive at
\begin{align*}
 F_k^{(l,l)}=&\frac{(-1)^k}{2^k}\left|
\begin{array}{ccccc}
1&0&0&\cdots&0\\
\beta_{l-1}&-2I_{l,l}&-2I_{l,l+1}&\cdots & -2I_{l,l+k-1}\\
\beta_l&-2I_{l+1,l}&-2I_{l+1,l+1}&\cdots & -2I_{l+1,l+k-1}\\
\vdots&\vdots&\vdots&\ddots&\vdots\\
\beta_{l+k-2}&-2I_{l+k-1,l}&-2I_{l+k-1,l+1}&\cdots & -2I_{l+k-1,l+k-1}
\end{array}
\right|\\
=&\frac{(-1)^k}{2^k}\left|
\begin{array}{ccccc}
1&\beta_l&\beta_{l+1}&\cdots&\beta_{l+k-1}\\
\beta_{l-1}&J_{l-1,l}&J_{l-1,l+1}&\cdots & J_{l-1,l+k-1}\\
\beta_l&J_{l,l}&J_{l,l+1}&\cdots & J_{l,l+k-1}\\
\vdots&\vdots&\vdots&\ddots&\vdots\\
\beta_{l+k-2}&J_{l+k-2,l}&J_{l+k-2,l+1}&\cdots & J_{l+k-2,l+k-1}
\end{array}
\right|\\
=&\frac{(-1)^k}{2^k}\left|
\begin{array}{ccccc}
J_{l-1,l}&J_{l-1,l+1}&\cdots & J_{l-1,l+k-1}\\
J_{l,l}&J_{l,l+1}&\cdots & J_{l,l+k-1}\\
\vdots&\vdots&\ddots&\vdots\\
J_{l+k-2,l}&J_{l+k-2,l+1}&\cdots & J_{l+k-2,l+k-1}
\end{array}
\right|\\
&-\frac{(-1)^k}{2^{k}}\left|
\begin{array}{ccccc}
0&\beta_l&\beta_{l+1}&\cdots&\beta_{l+k-1}\\
-\beta_{l-1}&J_{l-1,l}&J_{l-1,l+1}&\cdots & J_{l-1,l+k-1}\\
-\beta_l&J_{l,l}&J_{l,l+1}&\cdots & J_{l,l+k-1}\\
\vdots&\vdots&\vdots&\ddots&\vdots\\
-\beta_{l+k-2}&J_{l+k-2,l}&J_{l+k-2,l+1}&\cdots & J_{l+k-2,l+k-1}
\end{array}
\right|,
\end{align*}
and
\begin{align*}
 F_k^{(l+1,l-1)}
=&\frac{(-1)^k}{2^k}\left|
\begin{array}{ccccc}
J_{l,l-1}&J_{l,l}&\cdots & J_{l,l+k-2}\\
J_{l+1,l-1}&J_{l+1,l}&\cdots & J_{l+1,l+k-2}\\
\vdots&\vdots&\ddots&\vdots\\
J_{l+k-1,l-1}&J_{l+k-1,l}&\cdots & J_{l+k-1,l+k-2}
\end{array}
\right|\\
&-\frac{(-1)^k}{2^{k}}\left|
\begin{array}{ccccc}
0&\beta_{l-1}&\beta_{l}&\cdots&\beta_{l+k-2}\\
-\beta_{l}&J_{l,l-1}&J_{l,l}&\cdots & J_{l,l+k-2}\\
-\beta_l&J_{l+1,l-1}&J_{l+1,l}&\cdots & J_{l+1,l+k-2}\\
\vdots&\vdots&\vdots&\ddots&\vdots\\
-\beta_{l+k-1}&J_{l+k-1,l-1}&J_{l+k-1,l}&\cdots & J_{l+k-1,l+k-2}
\end{array}
\right|,
\end{align*}
from which the Pfaffian expressions for $F_k^{(l,l)}$ and $F_k^{(l+1,l-1)}$ are immediately obtained based on the identities \eqref{det_pf_odd}--\eqref{det_pf_even}. Thus the proof is completed.
\end{proof}

\section{Hermite--Pad\'{e} approximation involving Pfaffian structures}\label{sec:HP_pf}
Generally, in the Hermite--Pad\'e approximation the object of approximation is a system of analytic functions. In this section, we propose a mixed Hermite--Pad\'{e} approximation problem of a specific Nikishin system of functions (of order 2), the unique solution of which is given by some polynomials admitting Pfaffian structures, which are actually equivalent to those called PSOPs in \cite{chang2018partial}. 

It seems that this Hermite--Pad\'{e} approximation problem has never been discussed before.  In fact, the approximation problem studied here arises when considering the second column of the jump matrix for the spectral problem on the real line (as opposed to the
problems in \cite{mohajer2012inverse}, which are based on the first and the third columns, or \cite{hone2009explicit} where the finite interval $[-1,1]$ is used) related to Novikov multipeakons. This will be explained in Section \ref{sec:NVpeakon}.

\subsection{Hermite--Pad\'{e} approximation with unique solution}
\begin{define}\label{HP_pf}
Given a positive Stieltjes measure $d\mu$ on $\mathbb{R}_+$, let 
\begin{equation*}
W(z)=\int \frac{d\mu(x)}{z-x},\qquad  Z(z)=\iint\frac{d\mu(x)d\mu(y)}{(z-x)(x+y)}.
\end{equation*}
For fixed positive integer $k$ less than the number of increasing point of $d\mu$, seek for polynomials $Q_k(z),P_k(z),\hat P_k(z)$ with degrees $\deg(Q_k)=k$, $\deg(P_k)=\deg(\hat P_k)=k-1$, respectively, such that, as $z\rightarrow\infty$, 
\begin{subequations}\label{eq:HP_pf}
\begin{align}
&Q_k(z)W(z)-P_k(z)=\mathcal{O}(1),\label{HP_pf1}\\
&Q_k(z)Z(z)-\hat P_k(z)=\mathcal{O}\left(\frac{1}{z}\right),\label{HP_pf2}\\
&\hat P_k(z)+P_k(z)W(-z)+Q_k(z)Z(-z)=\mathcal{O}\left(\frac{1}{z^{k+1}}\right).\label{HP_pf3}
\end{align}
In addition, it is required that $Q_k(0)=0,P_k(0)=1$.
\end{subequations}
\end{define}

The pair $(W(z),Z(z))$ above constitutes a specific Nikishin system of functions (of order 2) and the approximation \eqref{eq:HP_pf} is a mixed type Hermite--Pad\'e approximation. The existence and uniqueness of the solution of this approximation problem are given in the following theorem. 

\begin{theorem}\label{th:sol_HP_pf1}
The Hermite--Pad\'{e} approximation problem \ref{HP_pf} has a unique solution as 
\begin{equation}\label{HP_pf_Q}
Q_k(z)=\frac{-1}{F_k^{(1,0)}}
\left|
\begin{array}{ccccc}
0&z&z^2&\cdots &z^k\\
\beta_0&I_{1,0}&I_{1,1}&\cdots & I_{1,k-1}\\
\beta_1&I_{2,0}&I_{2,1}&\cdots&I_{2,k-1}\\
\vdots&\vdots&\vdots&\ddots&\vdots\\
\beta_{k-1}&I_{k,0}&I_{k,1}&\cdots&I_{k,k-1}
\end{array}
\right|,
\end{equation}
and
\begin{align}
&P_k(z)=\int \frac{Q_k(z)-Q_k(x)}{z-x}d\mu(x)-\int \frac{Q_k(x)}{x}d\mu(x)+1,\label{HP_pf_P}\\
&\hat P_k(z)=\iint\frac{Q_k(z)-Q_k(x)}{(z-x)(x+y)}d\mu(x)d\mu(y),\label{HP_pf_hatP}
\end{align}
where $P_k(z)$ and $\hat P_k(z)$ are uniquely determined by $Q_k(z)$. Note that we use the notations $\beta_j$ and $I_{i,j}$ appearing at the beginning of Section \ref{sec:det_pf}.
\end{theorem}
\begin{proof}
First, once $Q_k(z)$ is known, we have
\begin{align*}
&Q_k(z)W(z)-P_k(z)=\int \frac{zQ_k(x)}{x(z-x)}d\mu(x)-1,\\
&Q_k(z)Z(z)-\hat P_k(z)=\iint\frac{Q_k(x)}{(z-x)(x+y)}d\mu(x)d\mu(y),
\end{align*}
which immediately give the first two approximation formulae \eqref{HP_pf1} and \eqref{HP_pf2}. The remainder of the proof is devoted to the existence and uniqueness of $Q(z)$ from the third approximation problem \eqref{HP_pf3}.

Substituting the expressions into \eqref{HP_pf3}, after simplification, we obtain
\[
\iint\frac{yQ_k(x)}{x(z+y)(x+y)}d\mu(x)d\mu(y)-\int\frac{1}{z+x}d\mu(x)=\mathcal{O}\left(\frac{1}{z^{k+1}}\right),
\]
which can be equivalently written as
\[
\iint\frac{yQ_k(x)}{x(x+y)}\sum_{j=0}^\infty\frac{(-y)^j}{z^{j+1}}d\mu(x)d\mu(y)-\int\sum_{j=0}^\infty\frac{(-x)^{j}}{z^{j+1}}d\mu(x)=\mathcal{O}\left(\frac{1}{z^{k+1}}\right).
\]
Obviously, the approximation order implies
\begin{equation}\label{orth_Q}
\iint\frac{y^{j+1}Q_k(x)}{x(x+y)}d\mu(x)d\mu(y)-\int x^jd\mu(x)=0, \quad j=0,1,\ldots,k-1.
\end{equation}
If we let $Q_k(z)=\sum_{i=1}^k c_iz^i, $ then the following linear system holds
\begin{equation*}
\left(
\begin{array}{cccc}
I_{1,0}&I_{1,1}&\cdots & I_{1,k-1}\\
I_{2,0}&I_{2,1}&\cdots&I_{2,k-1}\\
\vdots&\vdots&\ddots&\vdots\\
I_{k,0}&I_{k,1}&\cdots&I_{k,k-1}
\end{array}
\right)
\left(
\begin{array}{c}
c_1\\
c_2\\
\vdots\\
c_k
\end{array}
\right)
=
\left(
\begin{array}{c}
\beta_0\\
\beta_1\\
\vdots\\
\beta_{k-1}
\end{array}
\right).
\end{equation*}
Therefore $Q_k(z)$ is uniquely given by the formula \eqref{HP_pf_Q} since the coefficient matrix of the above linear system is nonsingular according to Property \ref{prop:FG}.
\end{proof}

\begin{remark}
The noteworthy point is that, the approximation condition \eqref{orth_Q} contains the term $\beta_j=\int x^jd\mu(x)$, which is the particular ingredient  responsible for the appearance of Pfaffian expressions in the following subsection.  This particular term has not appeared in the previously studied problems in \cite{hone2009explicit,mohajer2012inverse}.
\end{remark}

\subsection{The solution in terms of Pfaffians}
The unique solution of the Hermite--Pad\'{e} approximation problem \ref{HP_pf} is presented in terms of determinants \eqref{HP_pf_Q}, however, we would claim that it is more natural to express them in terms of Pfaffians. 
\begin{theorem}\label{th:sol_HP_pf2}
The Hermite--Pad\'{e} approximation problem \ref{HP_pf} has a unique solution. $Q_k(z)$ admits a Pfaffian expression  
\begin{equation}\label{HP_pf_Q_pf}
Q_k(z)=
\left\{
\begin{array}{ll}
\frac{-2z}{\tau_{2l}}\Pf \left(\left(
\begin{array}{ccccc}
0&\vline&
\beta_j&\vline&0
\\
\hline
-\beta_i&\vline&J_{i,j}&\vline&z\\
\hline
0&\vline&-z&\vline&0
\end{array}
\right)_{i,j=0}^{2l-1}\right),&k=2l,\\
\ \\
\frac{2z}{\tau_{2l-1}} \Pf\left(\left(
\begin{array}{ccccc}
J_{i,j}&\vline&z\\
\hline
-z&\vline&0
\end{array}
\right)_{i,j=0}^{2l-2}\right)
,&k=2l-1,
\end{array}
\right.
\end{equation}
with the $\tau_k^{(j)}$ in Definition \ref{def:tau}.
Besides, $P_k(z)$ and $\hat P_k(z)$ are given by \eqref{HP_pf_P} and \eqref{HP_pf_hatP} that are uniquely determined by $Q_k(z)$. In particular, 
$$\hat P_k(0)=\frac{\tau_{1}^{(-1)}}{\tau_{0}^{(0)}}-\frac{\tau_{k+1}^{(-1)}}{\tau_{k}^{(0)}}=\beta_{-1}-\frac{\tau_{k+1}^{(-1)}}{\tau_{k}^{(0)}}.$$
\end{theorem}
\begin{proof}
First, we rewrite $Q_k(z)$ as
\begin{align*}
Q_k(z)=\frac{(-1)^{k}}{2^{k-1}F_k^{(1,0)}}
\left|
\begin{array}{ccccc}
0&z&z^2&\cdots &z^k\\
\beta_0&-2I_{1,0}&-2I_{1,1}&\cdots & -2I_{1,k-1}\\
\beta_1&-2I_{2,0}&-2I_{2,1}&\cdots&-2I_{2,k-1}\\
\vdots&\vdots&\vdots&\ddots&\vdots\\
\beta_{k-1}&-2I_{k,0}&-2I_{k,1}&\cdots&-2I_{k,k-1}
\end{array}
\right|.
\end{align*}
For the above determinant, adding the first column multiplied by $\beta_{j-2}$ from the $j$-th column with $j=2,3\cdots,k+1$, we get
\begin{align*}
Q_k(z)=\frac{(-1)^{k}}{2^{k-1}F_k^{(1,0)}}
\left|
\begin{array}{ccccc}
0&z&z^2&\cdots &z^k\\
\beta_0&J_{0,0}&J_{0,1}&\cdots & J_{0,k-1}\\
\beta_1&J_{1,0}&J_{1,1}&\cdots&J_{1,k-1}\\
\vdots&\vdots&\vdots&\ddots&\vdots\\
\beta_{k-1}&J_{k-1,0}&J_{k-1,1}&\cdots&J_{k-1,k-1}
\end{array}
\right|,
\end{align*}
where the relations $J_{i,j}=\beta_i\beta_j-2I_{i+1,j}$ in \eqref{rel:IJb} are used. In order to get \eqref{HP_pf_Q_pf}, it is sufficient to employ the relations in\eqref{rel:FGtau} and the identities \eqref{det_pf_odd}--\eqref{det_pf_even}.

Regarding $\hat P_k(0)$, we have 
\begin{align*}
\hat P_k(0)&=\iint\frac{Q_k(x)}{x(x+y)}d\mu(x)d\mu(y)=\frac{-1}{F_k^{(1,0)}}
\left|
\begin{array}{ccccc}
0&I_{0,0}&I_{0,1}&\cdots & I_{0,k-1}\\
\beta_0&I_{1,0}&I_{1,1}&\cdots & I_{1,k-1}\\
\beta_1&I_{2,0}&I_{2,1}&\cdots&I_{2,k-1}\\
\vdots&\vdots&\vdots&\ddots&\vdots\\
\beta_{k-1}&I_{k,0}&I_{k,1}&\cdots&I_{k,k-1}
\end{array}
\right|\\
&=\beta_{-1}-\frac{1}{F_k^{(1,0)}}
\left|
\begin{array}{ccccc}
\beta_{-1}&I_{0,0}&I_{0,1}&\cdots & I_{0,k-1}\\
\beta_0&I_{1,0}&I_{1,1}&\cdots & I_{1,k-1}\\
\beta_1&I_{2,0}&I_{2,1}&\cdots&I_{2,k-1}\\
\vdots&\vdots&\vdots&\ddots&\vdots\\
\beta_{k-1}&I_{k,0}&I_{k,1}&\cdots&I_{k,k-1}
\end{array}
\right|=\beta_{-1}-\frac{G_{k+1}^{(-1,1)}}{F_k^{(1,0)}},
\end{align*}
where the relations $I_{i+1,j}+I_{i,j+1}=\beta_i\beta_j$ in \eqref{rel:IJb} are used in the last step. Eventually, the proof is completed by using the formulae in \eqref{rel:FGtau}.
\end{proof}

\subsection{Relation with PSOPs}
We conclude this section by pointing out the connection between the Pfaffian expressions of $Q_k(z)$ in \eqref{HP_pf_Q_pf} and a class of polynomials recently studied  in \cite{chang2018partial}, to which, a brief introduction is given below. 

\subsubsection{Brief review of PSOPs}

The new concept called partial-skew-orthogonal polynomials (PSOPs) was introduced in \cite{chang2018partial} based on a skew-symmetric inner product $\langle \cdot, \cdot\rangle$. It can be regarded as an intriguing modification of skew-orthogonal polynomials (SOPs), which are well-known in the random matrix theory because of the intimate connections with Gaussian orthogonal ensembles (GOE) and Gaussian symplectic ensembles (GSE). Actually, the introduction of PSOPs was partially motivated by the characteristic polynomials of the Bures ensemble appeared in \cite{forrester2016relating}. For the general setting of PSOPs in \cite{chang2018partial}, one can state as follows.

\begin{define}Given a bilinear 2-form from $\mathbb{R}(z)\times\mathbb{R}(z)\rightarrow \mathbb{R}$ satisfying the skew symmetric relation:
$$\langle f(z),g(z)\rangle=-\langle g(z),f(z)\rangle,$$
a family of monic polynomials $\{P_k(z)\}_{k=0}^{\infty}$ are called monic partial-skew-orthogonal polynomials (PSOPs)  with respect to the  skew-symmetric inner product $\langle \cdot, \cdot\rangle$ if they satisfy the following relations:
\begin{align*}
&\langle P_{2l}(z),z^j\rangle=\alpha_l\delta_{2l+1,j},\qquad 0\leq j\leq 2l-1,\\
&\langle P_{2l+1}(z),z^j\rangle=\eta_j\gamma_l,\quad \qquad 0\leq j\leq 2l+1,
\end{align*}
for some appropriate nonzero numbers $\alpha_l,\eta_j,\gamma_l$.
\end{define}

The concrete case concerned in \cite{chang2018partial} is the monic PSOP sequence
\begin{equation*}
P_k(z)=
\left\{
\begin{array}{ll}
\frac{1}{\tau_{2l}}\Pf\left(\left(
\begin{array}{ccccc}
\langle z^i,z^j\rangle&\vline&z\\
\hline
-z&\vline&0
\end{array}
\right)_{i,j=0}^{2l}\right),&k=2l,\\
\ \\
\frac{1}{\tau_{2l-1}}\Pf \left(\left(
\begin{array}{ccccc}
0&\vline&
\eta_j&\vline&0
\\
\hline
-\eta_i&\vline&\langle z^i,z^j\rangle&\vline&z\\
\hline
0&\vline&-z&\vline&0
\end{array}
\right)_{i,j=0}^{2l-1}\right), &k=2l-1.
\end{array}
\right.
\end{equation*}
Different settings for $\{\langle z^i,z^j\rangle,\eta_j\}$ have been investigated and the isospectral deformations can lead to some new and old integrable lattices \cite{chang2018partial}.

In particular,  if one relates the skew-symmetric product $\langle \cdot,\cdot\rangle$ with the bimoment $J_{i,j}$, and $\eta_j$ with the single moment $\beta_j$, i.e.
$$
\langle z^i,z^j\rangle\triangleq J_{i,j}=\iint \frac{y-x}{x+y}x^iy^jd\mu(x)d\mu(y),\quad \eta_j=\beta_j=\int x^jd\mu(x),
$$
then the specific PSOPs become
\begin{equation}\label{psop}
P_k^B(z)=
\left\{
\begin{array}{ll}
\frac{1}{\tau_{2l}}\Pf\left(\left(
\begin{array}{ccccc}
J_{i,j}&\vline&z\\
\hline
-z&\vline&0
\end{array}
\right)_{i,j=0}^{2l}\right),&k=2l,\\
\ \\
\frac{1}{\tau_{2l-1}}\Pf \left(\left(
\begin{array}{ccccc}
0&\vline&
\beta_j&\vline&0
\\
\hline
-\beta_i&\vline&J_{i,j}&\vline&z\\
\hline
0&\vline&-z&\vline&0
\end{array}
\right)_{i,j=0}^{2l-1}\right), &k=2l-1.
\end{array}
\right.
\end{equation}
These PSOPs are nothing but the averages  of characteristic polynomials of Bures ensemble \cite{forrester2016relating} based on the facts that
\begin{align*}
{P}_{k}^B(z)=\frac{1}{\tau_k}\ \ \ \ \idotsint\limits_{0<x_1<\cdots<x_k}\prod_{1\leq i<j\leq k}\frac{(x_j-x_i)^2}{(x_i+x_j)}{\prod_{1\leq j\leq k}(z-x_j)}d \mu(x_1)\cdots \mu(x_k).
\end{align*}
It is noted that the specific PSOPs satisfy the ``partial-skew-orthogonality'' condition: 
 \begin{equation}\label{orth:psop}
\begin{aligned}
&\iint P^B_{2l}(x)y^j\frac{y-x}{x+y}d\mu(x)d\mu(y)=\frac{\tau_{2l+2}}{\tau_{2l}}\delta_{2l+1,j},\qquad 0\leq j\leq 2l+1,\\
& \iint P^B_{2l-1}(x)y^j\frac{y-x}{x+y}d\mu(x)d\mu(y)=-\frac{\tau_{2l}}{\tau_{2l-1}}\beta_j,\qquad 0\leq j\leq 2l-1.
\end{aligned}
\end{equation}
Furthermore, it deserves to remark that the specific PSOPs admit a four-term recurrence in the form of 
\begin{equation*}
z(P_k^B-v_kP_{k-1}^B)=P_{k+1}^B+(d_{k}-v_k)P_k^B+v_k(d_k-v_{k+1})P_{k-1}^B+(v_k)^2v_{k-1}P_{k-2}^B,
\end{equation*}
which can act as the spatial part for an isospectral problem of the B-Toda lattice (see Subsection \ref{subsec:cbop_psop} for details).

\subsubsection{Correspondence}
By comparing the expressions \eqref{HP_pf_Q_pf} and \eqref{psop}, it is not hard to see that $\{Q_k(z)\}_{k=1}^\infty$ is a class of specific PSOPs but not monic. More exactly, we have the following theorem.
\begin{theorem}There exists a correspondence between the unique solution $\{Q_k(z)\}_{k=1}^\infty$ for the Hermite--Pad\'{e} approximation problem \ref{HP_pf} and the monic PSOPs $\{P^B_{k}(z)\}_{k=0}^\infty$, that is
\begin{align}\label{rel:QP}
Q_k(z)=2(-1)^{k-1}\frac{\tau_{k-1}}{\tau_k}zP_{k-1}^{B}(z).
\end{align}
\end{theorem}
Combining  \eqref{orth_Q}, \eqref{orth:psop} and the relation \eqref{rel:QP}, we immediately have the following corollary.
\begin{coro}The approximation condition \eqref{orth_Q} satisfied by $\{Q_k(z)\}_{k=1}^\infty$ is equivalent to the partial-skew-orthogonality 
\begin{subequations}
\begin{align*}
&\iint \frac{Q_{2l+1}(x)}{x}y^j\frac{y-x}{x+y}d\mu(x)d\mu(y)=2\frac{\tau_{2l+2}}{\tau_{2l+1}}\delta_{2l+1,j},\qquad 0\leq j\leq 2l+1,\\
& \iint \frac{Q_{2l}(x)}{x}y^j\frac{y-x}{x+y}d\mu(x)d\mu(y)=2\int x^j d\mu(x),\qquad 0\leq j\leq 2l-1.
\end{align*}
\end{subequations}
\end{coro}


\section{Novikov multipeakons}\label{sec:NVpeakon}
The Hermite--Pad\'{e} approximation proposed in the above section naturally comes from an inverse problem for solving multipeakon solutions of the Novikov equation.  In this section, we employ this approximation problem together with its unique solution in terms of Pfaffians to construct the explicit forms of the Novikov multipeakons. It turns out that this approach is more natural and convenient.

As a cubic generalization of the famous CH equation, the Novikov equation
 \begin{equation}\label{eq:NV}
m_t+m_xu^2+3muu_x=0,\qquad m=u-u_{xx},
\end{equation}
was firstly derived by V. Novikov \cite{novikov2009generalisations} in a symmetry classification of nonlocal partial differential equations and firstly published in the paper by Hone and Wang \cite{hone2008integrable}, who found its Lax integrability as follows:
 \begin{equation}
D_x \Phi=U
\Phi, \qquad D_t \Phi=V\Phi\label{NV_Lax}
\end{equation}
with 
\begin{equation*}
U=\begin{pmatrix}
0&\lambda m &1\\
0&0&\lambda m\\
1&0&0
\end{pmatrix},\quad
V=\begin{pmatrix}
-uu_x&{\lambda }^{-1}{u_x}-\lambda u^2m& u_x^2\\
{\lambda }^{-1}{u}&-\lambda^{-2} &-{\lambda }^{-1}{u_x}-\lambda u^2m\\
-u^2&{\lambda }^{-1}{u}& uu_x
\end{pmatrix}.
\end{equation*}
 In other words, the compatibility condition 
$$(D_xD_t- D_tD_x )\Phi= 0$$ 
implies the zero curvature condition
$$ U_t-V_x+[U,V]=0,$$
which is exactly the Novikov equation.

The Novikov equation  \eqref{eq:NV} 
admits the multipeakon solution of the form
\begin{equation}\label{form:NVpeak}
u=\sum_{k=1}^N m_k(t)e^{-|x-x_k(t)|}
\end{equation}
in some weak sense if the positions and momenta satisfy the following ODE system:
\begin{align}
\dot x_{k}&=u(x_k)^2, \qquad \dot m_{k}=-m_ku(x_k)\langle u_x\rangle (x_k), \qquad 1\leq k \leq N,\label{NV_eq:peakon} 
\end{align}
where $\langle f \rangle(a)$ denotes the average of left and right limits at the point $a$.

By virtue of the Lax integrability, the ODE system for certain initial values can be explicitly solved by use of inverse spectral method, hence leading to the pure multipeakon solutions of the Novikov equation  \cite{hone2009explicit,mohajer2012inverse}. However, there are only determinant structures rather than Pfaffians appearing there. In what follows, we present an alternative process involving Pfaffian structures for the inverse problem to achieve the goal. It is based on the work of Mohajer and Szmigielski in \cite{mohajer2012inverse}, in which the analysis 
is carried out on the real axis.

\subsection{Forward problem} This subsection is a summary of facts from \cite{hone2009explicit,mohajer2012inverse}. In the peakon sector, without loss of generality, we assume the positions are ordered $-\infty=x_0<x_1<x_2<\cdots x_N<x_{N+1}=+\infty$ 
and all amplitudes $m_k$ are positive (which remains true for a while around $t=0$). 

By considering the $x$-part of the Lax pair in the peakon case, we have
\begin{equation*}
\begin{pmatrix}
\phi_1\\
\phi_{2}\\
\phi_3
\end{pmatrix}=
\begin{pmatrix}
A_ke^x+\lambda^2C_ke^{-x}\\
2\lambda B_{k}\\
A_ke^x-\lambda^2C_ke^{-x}
\end{pmatrix},\qquad x\in (x_k,x_{k+1}),
\end{equation*}
where the factor containing $\lambda$ is introduced for convenience.
As explained in \cite{hone2009explicit}, $\phi_3$ must be continuous and $\phi_1$ and $\phi_{2}$ are allowed to have jump discontuities. By considering the jump point $x_k$, it is not hard to obtain 
\begin{equation}\label{trans_ABC}
\begin{pmatrix}
A_k\\
B_{k}\\
C_k
\end{pmatrix}=
S_k(z)
\begin{pmatrix}
A_{k-1}\\
B_{k-1}\\
C_{k-1}
\end{pmatrix},
\end{equation}
where
\begin{equation*}
S_k(z)=
\begin{pmatrix}
1-zm_k^2&-2zm_ke^{-x_k}&-z^2m_k^2e^{-2x_k}\\
m_ke^{x_k}&1&zm_ke^{-x_k}\\
m_k^2e^{2x_k}&2m_ke^{x_k}&1+zm_k^2
\end{pmatrix}, \qquad z=-\lambda^2.
\end{equation*}

Set the boundary condition $(A_0,B_0,C_0)=(1,0,0)$, then $(A_k,B_k,C_k)$ will be used for inverse problem. By induction, one can easily get some properties of $(A_k,B_k,C_k)$ from the explicit form of the transition matrix $S_k$. For convenience, we make use of the notation $q[j]$ for any polynomial $q(z)$ denoting its coefficient of $z^j$.
\begin{prop}
For $1\leq k\leq N$, we have 
\begin{enumerate}[(i)]
\item $\deg(A_k) = k, \deg(B_k) = k-1, \deg(C_k) = k-1$;
\item $A_k[0]=1$, $B_k[0]=\sum_{j=1}^km_je^{x_j}$.
\end{enumerate}
\end{prop}

Furthermore, the boundary condition $(A_0,B_0,C_0)=(1,0,0)$ is consistent with the time evolution for $x<x_1$. For $x>x_N$, the $t$-part of the Lax pair implies the evolution
\begin{equation}\label{evo:ABC}
\dot A_N=0,\quad  \dot B_N=\frac{B_N-A_NM_{+}}{z},\quad \dot C_N=\frac{2M_{+}(B_N-A_NM_{+})}{z},
\end{equation}
where $M_+=\sum_{j=1}^Nm_je^{x_j}$. This means that we can define a boundary value problem
\begin{equation} \label{bvp}
D_x \Phi=\begin{pmatrix}
0&\lambda m &1\\
0&0&\lambda m\\
1&0&0
\end{pmatrix}
\Phi, \qquad \Phi(-\infty)=\Phi_x(-\infty)=0, \Phi_x(+\infty)=0,
\end{equation}
which is consistent with the Lax pair \eqref{NV_Lax}. The spectrum of the boundary value problem is the set of all eigenvalues, i.e. the roots of $A_N(z)=0$.

By making use of Krein's theory of oscillatory matrices, one can show the following result.
\begin{theorem}
The spectrum of the boundary value problem \eqref{bvp} is positive and simple. Equivalently, $A_N(z)=\prod_{j=1}^N(1-\frac{z}{\zeta_j})$ with $0<\zeta_1<\zeta_2<\cdots <\zeta_N$.
\end{theorem}
In fact, not only the eigenvalues but also some positive constants known in scattering theory as the norming constants can be obtained from the boundary value problem. They are encoded in the Weyl functions $W(z),Z(z)$ of the boundary value problem, which will play crucial roles in the formulation of the inverse problem.

\begin{theorem} \label{th_WZ}
The Weyl function $W(z)=-\frac{B_N(z)}{A_N(z)}$ is a Stieltjes transform of a positive, discrete measure $d\mu$ with support inside $\mathbb R_+$, that is,
\begin{align}\label{weyl1}
W(z)=\int \frac{d\mu(x)}{z-x},\qquad & d\mu(x)=\sum_{j=1}^Nb_j\delta(x-\zeta_j)dx,
\end{align}
with $0<\zeta_1<\zeta_2<\cdots <\zeta_N, \ b_j>0.$
Moreover, the second Weyl function $Z(z)=-\frac{C_N(z)}{2A_N(z)}$ can be represented as 
\begin{align}\label{weyl2}
Z(z)=\int \frac{d\mu(x)d\mu(y)}{(z-x)(x+y)},
\end{align}
which is determined by the first Weyl function $W(z)$. Note that $W(z)$ and $Z(z)$ constitute a Nikishin system with respect to the measures $d\mu(x)$ and $d\mu(-x)$.
\end{theorem}

The spectrum $\{\zeta_k\}_{k=1}^N$ and the residues $\{b_k\}_{k=1}^N$ form the spectral data $\{\zeta_k,b_k\}_{k=1}^N$.  Since $\dot A_N=0$, the spectrum is invariant. With the help of \eqref{evo:ABC}, one can also obtain the time evolution of the residues $\{b_k\}_{k=1}^N$. 
\begin{theorem}
The spectral data $\{\zeta_k,b_k\}_{k=1}^N$ evolve according to 
\begin{align}\label{evo:zb}
\dot\zeta_k=0,\qquad \dot b_k=\frac{b_k}{\zeta_k}.
\end{align}
\end{theorem}

\begin{remark} It is noted that our notations are slightly different from those in the previous references. The Weyl functions $(W, Z)$ in Theorem \ref{th_WZ}  coincide with those in \cite{mohajer2012inverse}, but do not agree with \cite{hone2009explicit}, where they are called $(\omega,\zeta)$. Also note that the roles of $z$ and $\lambda$ are interchanged by comparing with our paper with \cite{hone2009explicit,mohajer2012inverse}, the reason of which is that $z$ is used  as a polynomial independent variable so that it is convenient to describe the relations between the entries of the transition matrix and CBOPs or PSOPs the following section.
Furthermore, with the spectral variable $z$ in place, we denote discrete spectral variables by $\zeta_k$, previously called $\lambda_k$ in\cite{hone2009explicit,kohlenberg2007inverse,mohajer2012inverse}. 
\end{remark}

We end this subsection by stating two properties on the transition matrix $S_k$.
First, the transition matrix $S_k$ admits certain symmetry properties \cite{mohajer2012inverse}, which will play important roles in the formulation of the  approximation problems.
\begin{lemma} \label{lem:sym_S}
The transition matrix $S_k$ satisfies
\begin{subequations}
\begin{align*}
&\det(S_k(z))=1,\\
& K^{-1}S_k(z)^TK = S_k(-z), \quad S_k(z)^{-1} =JS_k(z)J,\\
&(KJ)^{-1}(S_k(-z)^T )^{-1}KJ = S_k(z),
\end{align*}
\end{subequations}
where 
\begin{equation*}
K=\begin{pmatrix}
0&0&1\\
0&2&0\\
1&0&0
\end{pmatrix},\qquad 
J=\begin{pmatrix}
1&0&0\\
0&-1&0\\
0&0&1
\end{pmatrix}.
\end{equation*}
\end{lemma}
Moreover, as shown in \cite{mohajer2012inverse}, one can explicitly evaluate the degrees of the partial  transition matrix
\begin{equation*}
S_{[N,k]}(z):=S_N(z)S_{N-1}(z)\cdots S_{N-k+1}(z),
\end{equation*}
and its adjoint by induction.
It is noted that one has
\begin{align}\label{eq:S_N-k}
(A_N, B_{N},C_N)^T=
S_{[N,k]}(z)
(A_{N-k},B_{N-k},C_{N-k})^T
\end{align}
by use of \eqref{trans_ABC}.
\begin{lemma}\label{lem:S-N-k_deg}
The degrees of the corresponding entries of $S_{[n,k]}(z)$ and its adjoint are as follows:
\begin{equation*}
\deg(S_{[n,k]}(z))=\deg(S_{[n,k]}(z)^*)=
\begin{pmatrix}
k&k&k+1\\
k-1&k-1&k\\
k-1&k-1&k
\end{pmatrix}.
\end{equation*}
\end{lemma}


\subsection{Inverse problem}
The inverse problem associated with the boundary value problem \eqref{bvp} can be stated as follows:
Given the spectral data consisting of $\{\zeta_j,b_j\}_{j=1}^N$ (or equivalently, given a rational function $W(z)$ with integral representation \eqref{weyl1}), the problem is to recover the positions and masses $\{x_j,m_j\}_{j=1}^N$ so that the Weyl functions of the boundary value problem \eqref{bvp} have $W(z)$ in conjunction with $Z(z)$ as its Weyl functions. This subsection is devoted to 
solving the inverse problem by using a new approach.

The main technique for the solution to the inverse problem is to make use of simultaneous Hermite--Pad\'{e} approximations, however,  our solution is based on the Hermite--Pad\'{e} approximation problem using the second column of the transition matrix $S_{[n,k]}(z)$. This differs from the approaches in \cite{hone2009explicit,mohajer2012inverse}, where the first and third columns of the transition matrix on the real line are considered  in \cite{mohajer2012inverse}, while the finite interval $[-1,1]$ together with the corresponding first column is used in \cite{hone2009explicit}. 

\begin{theorem}
Let $(q_k(z),p_k(z),\hat p_k(z))^T$ denote the second column of the transition matrix $S_{[n,k]}(z)$. Then, for fixed $k$ and $z\rightarrow\infty$,
\begin{subequations}\label{appr_2col}
\begin{align*}
&\frac{B_N(z)}{A_N(z)}-\frac{p_k(z)}{q_k(z)}=\mathcal{O}(z^{-k}),\\
&\frac{C_N(z)}{A_N(z)}-\frac{\hat p_k(z)}{q_k(z)}=\mathcal{O}(z^{-k-1}),\\
&\hat p_k(-z)-2p_k(-z)\frac{B_N(z)}{A_N(z)}+q_k(-z)\frac{C_N(z)}{A_N(z)}=\mathcal{O}(z^{-k-1}),
\end{align*}
where 
\begin{equation*}
\deg(q_k)=k,\quad \deg(p_k)=\deg(\hat p_k)=k-1,\quad q_k(0)=0,\quad p_k(0)=1.
\end{equation*}
\end{subequations}
In addition, 
\begin{equation}\label{rel:qp_mx}
q_k[1]=-2\sum_{j=N-k+1}^Nm_je^{-x_j},\quad \hat p_k[0]=2\sum_{j=N-k+1}^Nm_je^{x_j},
\end{equation}
where, again, the notation $q[j]$ for any polynomial $q(z)$ denotes its coefficient of $z^j$. 
\end{theorem}
\begin{proof}
The desired specific values or degrees of $(q_k(z),p_k(z),\hat p_k(z))^T$ easily follow from induction or directly from the previous Lemma. The reminder of the proof is devoted to the three approximation equations.

Let $s_{ij}$ denote the entries of $S_{[n,k]}(z)$. It follows from \eqref{eq:S_N-k} that 
\begin{align*}
\frac{B_N(z)}{A_N(z)}-\frac{p_k(z)}{q_k(z)}=&\frac{s_{21}A_{N-k}+s_{22}B_{N-k}+s_{23}C_{N-k}}{s_{11}A_{N-k}+s_{12}B_{N-k}+s_{13}C_{N-k}}-\frac{s_{22}}{s_{12}}\\
=&\frac{\left(s_{12}s_{21}-s_{11}s_{22}\right)\frac{A_{N-k}}{B_{N-k}}+\left(s_{12}s_{23}-s_{13}s_{22}\right)\frac{C_{N-k}}{B_{N-k}}}{s_{11}\left(s_{11}+s_{12}\frac{A_{N-k}}{B_{N-k}}+s_{13}\frac{C_{N-k}}{B_{N-k}}\right)}.
\end{align*}
 Observe the fact that $s_{12}s_{21}-s_{11}s_{22}$ and $s_{12}s_{23}-s_{13}s_{22}$ have the same degrees as the $(3,3)$ and $(1,3)$ entries of the adjoint of $S_{[n,k]}(z)$. Then, by exploiting Lemma \ref{lem:S-N-k_deg}, one can obtain the first approximation equation, namely the one for $\frac{B_N(z)}{A_N(z)}$. Similar argument applies for the case for $\frac{C_N(z)}{A_N(z)}$.
 In order to obtain the third approximation equation, we observe the formula
 \begin{equation*}
 (A_N,B_N,C_N)KJS_{[N,k]}(-z)=(A_{N-k},B_{N-k},C_{N-k})KJ,
 \end{equation*}
 which can be produced by \eqref{eq:S_N-k} and the identities in Lemma \ref{lem:sym_S}. Checking the second column of the above equality gives 
 $$
\hat p_k(-z)-2p_k(-z)\frac{B_N(z)}{A_N(z)}+q_k(-z)\frac{C_N(z)}{A_N(z)}=-2\frac{B_{N-k}(z)}{A_N(z)}=\mathcal{O}(z^{-k-1}),
 $$
 which completes the proof.
\end{proof}

Notice that \eqref{appr_2col} matches perfectly with our proposed Hermite--Pad\'{e} approximation model \ref{HP_pf}. Then it follows from Theorem \ref{th:sol_HP_pf1} (as well as Theorem \ref{th:sol_HP_pf2}) that the second column $(q_k(z),p_k(z),\hat p_k(z))^T$ of the transition matrix $S_{[n,k]}(z)$ can be uniquely recovered for given rational functions $W(z)$ and $Z(z)$ in  \eqref{weyl1}--\eqref{weyl2} (a specific Nikishin system with respect to the spectral measure $d\mu$). More exactly, we have the following theorem.
\begin{theorem}\label{th:WZ_pq}
 The second column $(q_k(z),p_k(z),\hat p_k(z))^T$ of the transition matrix $S_{[n,k]}(z)$ can be uniquely recovered from the spectral data and they are explicitly given by 
\begin{align*}
q_k(z)=-Q_k(z),\quad p_k(z)=P_k(z),\quad \hat p_k(z)=2\hat P_k(z)
\end{align*}
where $(Q_k,P_k,\hat P_k)$ are given by Pfaffian expressions appearing in Theorem \ref{th:sol_HP_pf2}.
\end{theorem}
 Once the second column $(q_k(z),p_k(z),\hat p_k(z))^T$ is known, $\{x_k,m_k\}_{k=1}^N$ can be recovered and they admit expressions in terms of Pfaffians, which are presented in the following subsection.
\subsection{Solutions in terms of Pfaffians}
 In fact, it is obvious that the formulae in \eqref{rel:qp_mx} imply the following result.
\begin{coro}\label{coro:pq_xm}
 For any integer $1\leq k\leq N$, we have
\begin{subequations}
\begin{align*}
-2m_{N-k+1}e^{-x_{N-k+1}}&=q_k[1]-q_{k-1}[1],\\
2m_{N-k+1}e^{x_{N-k+1}}&=\hat p_k[0]-\hat p_{k-1}[0],
\end{align*}
\end{subequations}
with the convention $q_0[1]=\hat p_0[0]=0$.
\end{coro}

Combing Theorems \ref{th:sol_HP_pf2}, \ref{th:WZ_pq} and Corollary \ref{coro:pq_xm}, we finally get the inverse map from the spectral data $\{\zeta_k,b_k\}_{k=1}^N$ to $\{x_k,m_k\}_{k=1}^N$.
\begin{theorem}\label{th:sol_xm}
Let $\{\zeta_k,b_k\}_{k=1}^N$ satisfying
$$0<\zeta_1<\zeta_2<\cdots <\zeta_N, \qquad b_j>0,$$ 
be the spectral data of the boundary value problem \eqref{bvp} with $m_k>0$ and $x_1<x_2<\cdots <x_N$. Then  $\{x_k,m_k\}_{k=1}^N$ can be uniquely recovered from the spectral data according to the Pfaffian formulae
\begin{equation}\label{form:xm}
\begin{aligned}
m_{N-k+1}e^{-x_{N-k+1}}&=\frac{\tau_{k-1}^{(1)}}{\tau_{k}}-\frac{\tau_{k-2}^{(1)}}{\tau_{k-1}},\\
m_{N-k+1}e^{x_{N-k+1}}&=\frac{\tau_{k}^{(-1)}}{\tau_{k-1}}-\frac{\tau_{k+1}^{(-1)}}{\tau_{k}},
\end{aligned}
\end{equation}
where $\tau_k^{(j)}$ denote Pfaffians in Definition \ref{def:tau} with the discrete measure $d\mu(x)=\sum_{j=1}^Nb_j\delta(x-\zeta_j)dx$.
\end{theorem}
\begin{remark}
The explicit formulae of Novikov's multipeakon consist of \eqref{form:NVpeak}, \eqref{evo:zb} and \eqref{form:xm}, which are consistent with the results in \cite{chang2018application,hone2009explicit,mohajer2012inverse}. Recall that the formulae in \cite{hone2009explicit,mohajer2012inverse} are given in terms of determinants, while Pfaffians are introduced in \cite{chang2018application} by observing some relations between Pfaffians and explicit formulae in \cite{hone2009explicit}. The main innovation here is that the Pfaffian formulae are naturally obtained from one Hermite--Pad\'{e} approximation problem with a family of PSOPs introduced in \cite{chang2018partial} as its unique solution. Furthermore, it is noted that the solution exists globally in time as shown in \cite{hone2009explicit}.  In other words, if the initial data are in $\mathcal{P}=\{x_1<x_2<\cdots<x_N,\  m_k>0,\   k=1,2,\ldots,N\}$, then $\{x_k(t),m_k(t)\}_{k=1}^N$ will exist for all the time $t\in \bf R$ under the peakon flow \eqref{NV_eq:peakon},
\end{remark}

\section{Dual cubic string}\label{sec:dual_cubic}
The pure multipeakon solutions for the Novikov equation was firstly obtained using inverse spectral method by Hone, Lundmark and Szmigielski in \cite{hone2009explicit},  by transforming  the spectral problem on the real axis into a dual cubic string in the finite interval $[-1,1]$. Then the inverse problem for the dual cubic string involving  three Hermite--Pad\'e approximation problems can be formulated. In this section, we shall investigate applications of the corresponding Hermite--Pad\'e approximation problems together with their solutions, from which some connections between CBOPs and PSOPs, as well as their corresponding integrable lattices, are established.

\subsection{Forward problem} This subsection is on the forward spectral problem, which is a summary of the facts from \cite{hone2009explicit}.
 As for \eqref{NV_Lax}, under the Liouville transformation,
\begin{equation*}
\begin{aligned}
y&=\tanh x,\\
\psi_1(y)&=\phi_1(x)\cosh x-\phi_3(x)\sinh x,\\
\psi_{2}(y)&=\lambda\phi_{2}(x),\\
\psi_3(y)&=\lambda^2\phi_3(x)/\cosh x,\\
g(y)&=m(x)\cosh^3 x,\\
z&=-\lambda^2,
\end{aligned}
\end{equation*}
one can obtain the dual cubic string
\begin{equation}\label{eq:dualcubic}
D_y\begin{pmatrix}
\psi_1\\
\psi_{2}\\
\psi_3
\end{pmatrix}=
\begin{pmatrix}
0& g(y)&0\\
0&0&g(y)\\
-z&0&0
\end{pmatrix}\begin{pmatrix}
\psi_1\\
\psi_{2}\\
\psi_3
\end{pmatrix},\qquad y\in(-1,1),
\end{equation}
with the boundary
\begin{equation*}
\psi_{2}(-1)=\psi_{3}(-1)=0,\quad \psi_{3}(1)=0.
\end{equation*}
Here we remark that the boundary condition is chosen appropriately to ensure 
isospectrality under specific deformations. 

In the peakon sector $m=2\sum_{i=1}^Nm_i\delta(x-x_i)$, one gets
\begin{equation}\label{map:xm_yg}
g=\sum_{i=1}^Ng_{i}\delta(x-y_i),\qquad g_{i}=2m_i\cosh {x_i},\quad y_i=\tanh x_i,
\end{equation}
which leads to the discrete dual cubic string problem. 

In the discrete case, it is reasonable to set 
\begin{equation*}
\begin{pmatrix}
\psi_1\\
\psi_{2}\\
\psi_3
\end{pmatrix}=
\begin{pmatrix}
A_k(z)-zC_k(z)\\
-2z B_{k}(z)\\
-zA_k(z)(1+y)-z^2C_k(z)(1-y)
\end{pmatrix},\qquad y\in (y_k,y_{k+1}).
\end{equation*}
The initial value $(\psi_1(-1),\psi_{2}(-1),\psi_3(-1))=(1,0,0)$ is equivalent to $(A_0,B_{0},C_0)=(1,0,0)$. At the end point $y=1$, we have 
\begin{equation*}
\begin{pmatrix}
\psi_1(1;z)\\
\psi_{2}(1;z)\\
\psi_3(1;z)
\end{pmatrix}=
\begin{pmatrix}
A_N(z)-zC_N(z)\\
-2z B_{N}(z)\\
-2zA_N(z)
\end{pmatrix},
\end{equation*}
from which we see that $A_N(z)=0$ defines the spectrum corresponding to $\psi_3(1;z)=0$ (except an additional eigenvalue $z_0=0$). Moreover, it follows that $\psi_3$ is continuous and piecewise linear, while $\psi_1,\psi_{2}$ are piecewise constants with jumps at $y_k$. More exactly, we have 
\begin{equation*}
\begin{aligned}
&\psi_1(y_k+)-\psi_1(y_k-)=g_{k}\avg{\psi_{2}}(y_k),\\
&\psi_{2}(y_k+)-\psi_{2}(y_k-)=g_{k}\psi_3(y_k),\\
&\psi_3(y_{k+1}-)-\psi_3(y_k+)=-zl_k\psi_1(y_k+),
\end{aligned}
\end{equation*}
where $l_k=y_{k+1}-y_k$ with $y_0=-1$, $y_{N+1}=1$. Equivalently, we obtain the formula for the transition of $\Psi=(\psi_1,\psi_{2},\psi_3)^T$ from $y_{k}-$ to $y_{k+1}-$
\begin{equation*}
\Psi(y_k+)=G_k\Psi(y_k-),\qquad \Psi(y_{k+1}-)=L_k(z)\Psi(y_k+),
\end{equation*}
where
\begin{equation}\label{cubic:GL}
G_k=
\begin{pmatrix}
1&g_{k}&\frac{g_{k}^2}{2}\\
0&1&g_{k}\\
0&0&1
\end{pmatrix},\qquad 
L_k(z)=\begin{pmatrix}
1&0&0\\
0&1&0\\
-zl_k&0&1
\end{pmatrix}.
\end{equation}
Then it follows immediately that
\begin{equation*}
\Psi(1;z)=L_NG_N\cdots L_1G_1L_0
\begin{pmatrix}
1\\
0\\
0
\end{pmatrix}.
\end{equation*}
As shown in \cite{hone2009explicit}, in the case of all $g_{k}>0$ (equivalently, all $m_k>0$), the discrete dual cubic string has $N$ simple and positive eigenvalues together with a zero eigenvalue.  Besides, two Weyl functions encoding the spectral data can be introduced as
\begin{equation}\label{cubic:weyl}
\begin{aligned}
&\tilde W(z)=-\frac{\psi_2(1;z)}{\psi_3(1;z)}=\frac{-B_N(z)}{A_N(z)}=W(z),\\
&\tilde Z(z)=-\frac{\psi_1(1;z)}{\psi_3(1;z)}=\frac{A_N(z)-zC_N(z)}{2zA_N(z)}=\frac{1}{2z}+Z(z),
\end{aligned}
\end{equation}
where $W(z),Z(z)$ are those functions in \eqref{weyl1}--\eqref{weyl2}. Note that the Weyl function $\tilde Z(z)$ is determined by $\tilde W(z)$.

Therefore, the above procedure forms a forward spectral map for the discrete dual cubic string, i.e. from the initial data $\{y_k,g_k\}_{k=1}^N$ satisfying $-1<y_1<y_2<\cdots<y_N<1, \ g_k>0$ to the spectral data $\{\zeta_k,b_k\}_{k=1}^N$ satisfying $0<\zeta_1<\zeta_2<\cdots<\zeta_N, \ b_k>0$.

\begin{remark}
The notations $( \tilde W , \tilde Z )$ in \eqref{cubic:weyl} do not agree with those in \cite{hone2009explicit}, where they are called
 $(W,Z)$ instead.
\end{remark}

\subsection{Two Hermite--Pad\'{e} approximations} Before stating the inverse problem in the following subsection, we present some known facts on the products of the transition matrices \eqref{cubic:GL} and the related corresponding approximation problems,  which mainly come from \cite{kohlenberg2007inverse}.

 Let 
\begin{align}\label{exp:mat_a}
&a^{(k)}(z)=L_N(z)G_N\cdots L_{N-k+1}(z)G_{N-k+1}L_{N-k}(z).
\end{align}
Then some information of $a^{(k)}$ and its adjoint can be obtained by induction (see \cite[Lemma 4.1, Theorem 4.2, Proposition 4.3]{kohlenberg2007inverse}).
\begin{lemma} \label{lem:deg_coe}
Fixed $1\leq k \leq N$, the polynomial degree of the matrices $a^{(k)}$ and its adjoint are given by
\begin{align*}
&\deg a^{(k)}(z)=\deg (a^{(k)}(z))^*=
\begin{pmatrix}
k&k-1&k-1\\
k&k-1&k-1\\
k+1&k&k
\end{pmatrix}.
\end{align*}
Furthermore, for $0\leq k\leq N$, we have 
\begin{subequations}
\begin{align*}
a^{(k)}[0]
=
\begin{pmatrix}
1&\sum_{j=N-k+1}^Ng_j&\sum_{j=N-k+1}^N\frac{g_j^2}{2}\\\
0&1&\sum_{j=N-k+1}^Ng_j\\
0&0&1
\end{pmatrix}
\end{align*}
and
\begin{align*}
&a_{31}^{(k)}[1]=-\sum_{j=N-k}^Nl_j,
\end{align*}
\end{subequations}
where we recall that the notation $q[j]$ for any polynomial $q(z)$ denotes its coefficient of $z^j$.
\end{lemma}

It is known from \cite[Theorem 4.2]{kohlenberg2007inverse} that the first and second columns of the transition matrix $a^{(k)}(z)$ satisfy two Hermite--Pad\'e approximation problems, which we call ``Type A" and ``Type B" (rather than``Type I" and ``Type II" used in our introduction), respectively. And, Theorems 4.12 and 4.15 in  \cite{kohlenberg2007inverse} ensure that these two approximation problems have unique solutions. 
\begin{define}[Type A approximation problem]\label{def:approx1} Given two rational functions $\tilde W(z), \tilde Z(z)$ defined in \eqref{cubic:weyl},
for fixed $0\leq k\leq N$, seek for polynomials $Q,P,\hat P$ of degree $k+1,k,k$, respectively, satisfying
\begin{subequations}
\begin{align*}
&Q(z)\tilde W(z)+P(z)=\mathcal{O}(1),\\
&Q(z)\tilde Z(z)+\hat P(z)=\mathcal{O}(z^{-1}),\\
&Q(z)\tilde Z(-z)-P(z)\tilde W(-z)-\hat P(z)=\mathcal{O}(z^{-(k+1)}),
\end{align*}
as $z\rightarrow \infty$, and
\begin{equation*}
Q(0)=0,\quad P(0)=0,\quad \hat P(0)=1.
\end{equation*}
\end{subequations}
\end{define}

\begin{theorem}\label{th:approx1}
The Hermite--Pad\'{e} approximation problem \ref{def:approx1} has a unique solution as 
\begin{align}\label{approx1:Q}
Q(z)=\frac{-1}{F_{k+1}^{(0,0)}+\frac{1}{2}F_k^{(1,1)}}
\left|
\begin{array}{cccc}
z&z^2&\cdots&z^{k+1}\\
I_{10}&I_{11}&\cdots&I_{1k}\\
\vdots&\vdots&\ddots&\vdots\\
I_{k0}&I_{k1}&\cdots&I_{kk}
\end{array}
\right|
\end{align}
and $P(z),\hat P(z)$ are uniquely determined by $Q(z)$ according to
\begin{subequations}
\begin{align*}
&P(z)=\int \frac{Q(x)-Q(z)}{z-x}d\mu(x)+\int \frac{Q(x)}{x}d\mu(x),\\
&\hat P(z)=-\frac{1}{2} \frac{Q(z)}{z}+\iint \frac{Q(x)-Q(z)}{(z-x)(x+y)}d\mu(x)d\mu(y).
\end{align*}
\end{subequations}
These polynomials $(\hat P(z), P(z), Q(z))$ give the first column $(a_{11}^{(k)}(z),a_{21}^{(k)}(z),a_{31}^{(k)}(z))$. In particular, we have 
\begin{equation*}
a_{31}^{(k)}[1]=\frac{-F_k^{(1,1)}}{F_{k+1}^{(0,0)}+\frac{1}{2}F_k^{(1,1)}}.
\end{equation*}
\end{theorem}

\begin{define}[Type B approximation problem]\label{def:approx2} Given two rational functions $\tilde W(z), \tilde Z(z)$ defined in \eqref{cubic:weyl},
for fixed $1\leq k\leq N$, seek for polynomials $Q,P,\hat P$ of degree $k,k-1,k-1$, respectively, satisfying
\begin{subequations}
\begin{align*}
&Q(z)\tilde W(z)+P(z)=\mathcal{O}(1),\\
&Q(z)\tilde Z(z)+\hat P(z)=\mathcal{O}(z^{-1}),\\
&Q(z)\tilde Z(-z)-P(z)\tilde W(-z)-\hat P(z)=\mathcal{O}(z^{-(k+1)}),
\end{align*}
as $z\rightarrow \infty$, and
\begin{equation*}
Q(0)=0,\quad P(0)=1.
\end{equation*}
\end{subequations}
\end{define}

\begin{theorem}
The Hermite--Pad\'{e} approximation problem \ref{def:approx2} has a unique solution as 
\begin{equation*}
Q_k(z)=\frac{1}{F_k^{(1,0)}}
\left|
\begin{array}{ccccc}
0&z&z^2&\cdots &z^k\\
\beta_0&I_{1,0}&I_{1,1}&\cdots & I_{1,k-1}\\
\beta_1&I_{2,0}&I_{2,1}&\cdots&I_{2,k-1}\\
\vdots&\vdots&\vdots&\ddots&\vdots\\
\beta_{k-1}&I_{k,0}&I_{k,1}&\cdots&I_{k,k-1}
\end{array}
\right|,
\end{equation*}
and $P(z),\hat P(z)$ are uniquely determined by $Q(z)$ according to
\begin{subequations}
\begin{align*}
&P(z)=1+\int \frac{Q(x)-Q(z)}{z-x}d\mu(x)+\int \frac{Q(x)}{x}d\mu(x),\\
&\hat P(z)=-\frac{1}{2} \frac{Q(z)}{z}+\iint \frac{Q(x)-Q(z)}{(z-x)(x+y)}d\mu(x)d\mu(y).
\end{align*}
\end{subequations}
These polynomials $(\hat P(z), P(z), Q(z))$ give the second column $(a_{12}^{(k)}(z),a_{22}^{(k)}(z),a_{32}^{(k)}(z))$. 
\end{theorem}

As is shown, the solutions to the above approximation problems were written in terms of determinants in \cite{kohlenberg2007inverse}. At the end of this subsection, we express the solution to the approximation problem \ref{def:approx2} in terms of Pfaffians, the proof of which is similar to that for Theorem \ref{th:sol_HP_pf2}.

\begin{theorem}\label{th:approx2}
The Hermite--Pad\'{e} approximation problem \ref{def:approx2} has a unique solution as 
\begin{equation}\label{approx2_Q_pf}
Q_k(z)=
\left\{
\begin{array}{ll}
\frac{2z}{\tau_{2l}}\Pf \left(\left(
\begin{array}{ccccc}
0&\vline&
\beta_j&\vline&0
\\
\hline
-\beta_i&\vline&J_{i,j}&\vline&z\\
\hline
0&\vline&-z&\vline&0
\end{array}
\right)_{i,j=0}^{2l-1}\right),&k=2l\\
\ \\
\frac{-2z}{\tau_{2l-1}}\Pf\left(\left(
\begin{array}{ccccc}
J_{i,j}&\vline&z\\
\hline
-z&\vline&0
\end{array}
\right)_{i,j=0}^{2l-2}\right),&k=2l-1
\end{array}
\right.
\end{equation}
and $P(z),\hat P(z)$ are uniquely determined by $Q(z)$ according to
\begin{subequations}
\begin{align*}
&P(z)=1+\int \frac{Q(x)-Q(z)}{z-x}d\mu(x)+\int \frac{Q(x)}{x}d\mu(x),\\
&\hat P(z)=-\frac{1}{2} \frac{Q(z)}{z}+\iint \frac{Q(x)-Q(z)}{(z-x)(x+y)}d\mu(x)d\mu(y).
\end{align*}
\end{subequations}
These polynomials $(\hat P(z), P(z), Q(z))$ give the second column $(a_{12}^{(k)}(z),a_{22}^{(k)}(z),a_{32}^{(k)}(z))$. In particular,  we have 
\begin{equation*}
a_{12}^{(k)}[0]=\frac{\tau_{k-1}^{(1)}}{\tau_{k}^{(0)}}+\frac{\tau_{1}^{(-1)}}{\tau_{0}^{(0)}}-\frac{\tau_{k+1}^{(-1)}}{\tau_{k}^{(0)}}=\beta_{-1}+\frac{\tau_{k-1}^{(1)}}{\tau_{k}^{(0)}}-\frac{\tau_{k+1}^{(-1)}}{\tau_{k}^{(0)}}.
\end{equation*}
\end{theorem}
\begin{remark}We note that ``Type A" and ``Type B" approximation problems are designed to approximate two Weyl functions $\tilde W(z), \tilde Z(z)$ rather than $W(z), Z(z)$. Although there is a certain similarity between the ``Type B" approximation \ref{def:approx2} and  the problem \ref{HP_pf}, the Hermite--Pad\'e approximation problem \ref{HP_pf} is a more natural one than the ``Type B" approximation \ref{def:approx2}. 

\end{remark}

\subsection{Inverse problem} Given the spectral data $\{\zeta_k,b_k\}_{k=1}^N$ satisfying
\begin{equation*}
0<\zeta_1<\zeta_2<\cdots<\zeta_N,\qquad b_1,b_2,\ldots,b_N>0
\end{equation*}
(or equivalently, the Weyl function $\tilde W(z)$ \eqref{cubic:weyl}),
one can uniquely recover the dual cubic string data $\{y_k,g_k\}_{k=1}^N$, which gave rise to $\tilde W(z)$ as well as $\tilde Z(z)$. This goal can be achieved by investigating Hermite--Pad\'e approximation problems formulated by the transition matrix \eqref{cubic:GL}, which was originally done in \cite{hone2009explicit} by using the first column of the transition matrix. In this subsection,   we present an alternative way to solve the inverse problem. Our approach is based on the the first two columns of the transition matrix.

 According to Lemma \ref{lem:deg_coe}, it is not hard to derive the following theorem. 
\begin{theorem} \label{th_lga}
The dual cubic string data $\{g_k,l_k\}$ are related to the first and the second columns according to
\begin{subequations} 
\begin{align*}
&g_{N-k+1}=a^{(k)}_{12}[0]-a^{(k-1)}_{12}[0],\qquad\quad k=1,\ldots, N,\\
&l_{N-k}=a_{31}^{(k-1)}[1]-a_{31}^{(k)}[1], \qquad\quad\ k=0,1,\ldots, N,\\
&y_{N-k+1}=1+a_{31}^{(k-1)}[1], \qquad\qquad\quad\ k=1,\ldots, N,
\end{align*}
where we make use of the convention $a^{(0)}_{12}[0]=a_{31}^{(-1)}[1]=0$.
\end{subequations}

\end{theorem}


Combining Theorems \ref{th:approx1}, \ref{th:approx2} and \ref{th_lga}, we finally obtain the formulae for $\{g_k,l_k\}$ in terms of spectral data $\{\zeta_k,b_k\}_{k=1}^N$.
\begin{theorem} 
Let $\{\zeta_k,b_k\}_{k=1}^N$ satisfying
$$0<\zeta_1<\zeta_2<\cdots <\zeta_N, \qquad b_j>0,$$ 
be the spectral data of the discrete dual cubic string \eqref{eq:dualcubic} with $g_k>0$ and $-1=y_0<y_1<y_2<\cdots y_N<y_{N+1}=1$. Then  $\{y_k,g_k\}_{k=1}^N$ can be uniquely recovered from the spectral data according to the formulae 
\begin{subequations}
\begin{align*}
g_{N-k+1}&=\frac{\tau_{k-1}^{(1)}}{\tau_{k}^{(0)}}-\frac{\tau_{k+1}^{(-1)}}{\tau_{k}^{(0)}}-\frac{\tau_{k-2}^{(1)}}{\tau_{k-1}^{(0)}}+\frac{\tau_{k}^{(-1)}}{\tau_{k-1}^{(0)}}=\frac{2^k\left(F_{k}^{(0,0)}+\frac{1}{2}F_{k-1}^{(1,1)}\right)}{\tau_{k}^{(0)}\tau_{k-1}^{(0)}},\\
l_{N-k}&=\frac{\left(\tau_{k}^{(0)}\right)^4}{2^{2k}\left(F_{k+1}^{(0,0)}+\frac{1}{2}F_k^{(1,1)}\right)\left(F_{k}^{(0,0)}+\frac{1}{2}F_{k-1}^{(1,1)}\right)},\\
y_{N-k+1}&=\frac{F_{k+1}^{(0,0)}-\frac{1}{2}F_k^{(1,1)}}{F_{k+1}^{(0,0)}+\frac{1}{2}F_k^{(1,1)}},
\end{align*}
\end{subequations}
where $\tau_k^{(j)}$ denote Pfaffians in Definition \ref{def:tau} with the discrete measure $d\mu(x)=\sum_{j=1}^Nb_j\delta(x-\zeta_j)dx$. Note that $F_k^{(l,l)}$ admit Pfaffian representations indicated in Theorem \ref{prop:FGtau}.
\end{theorem}

\begin{remark}
By employing the variable transformation \eqref{map:xm_yg}, one easily recovers $\{x_k,m_k\}_{j=1}^N$, which agrees with the result in Theorem \ref{th:sol_xm} based on the analysis on the real line. Although similar analysis in the finite interval was done in \cite{hone2009explicit},  our derivation is slightly different from that presented in \cite{hone2009explicit}, where only the first column is involved. 
\end{remark}

\begin{remark}
This subsection is only a by-product, while our main purpose is to associate the first two columns with CBOPs and PSOPs by virtue of the dual cubic string problem together with its transition matrix, and to reveal certain connections of recently studied integrable lattices, which will be presented in the following subsection.
\end{remark}

\subsection{Relations with nonsymmetric CBOPs \& PSOPs \& integrable lattices} \label{subsec:cbop_psop}
%


We now observe that the expressions \eqref{approx1:Q} and \eqref{approx2_Q_pf} for the $(3,1)$ and $(3,2)$ entries of the transition matrices $a^{(k)}$, obtained from the discrete dual cubic string, are associated with the specific nonsymmetric CBOPs and PSOPs, respectively.
In fact,
\begin{align*}
a^{(k)}_{31}(z)=\frac{(-1)^{k}F_k^{(1,0)}z}{F_{k+1}^{(0,0)}+\frac{1}{2}F_k^{(1,1)}}P_k^{C}(z),\qquad\quad a^{(k)}_{32}(z)=2(-1)^{k}\frac{\tau_{k-1}}{\tau_k}zP_{k-1}^{B}(z),
\end{align*}
where $P_k^{B}(z)$ is the $k$-th monic PSOP in \eqref{psop}, i.e.
\begin{equation*}
P_k^B(z)=
\left\{
\begin{array}{ll}
\frac{1}{\tau_{2l}}\Pf\left(\left(
\begin{array}{ccccc}
J_{i,j}&\vline&z\\
\hline
-z&\vline&0
\end{array}
\right)_{i,j=0}^{2l}\right),&k=2l,\\
\ \\
\frac{1}{\tau_{2l-1}}\Pf \left(\left(
\begin{array}{ccccc}
0&\vline&
\beta_j&\vline&0
\\
\hline
-\beta_i&\vline&J_{i,j}&\vline&z\\
\hline
0&\vline&-z&\vline&0
\end{array}
\right)_{i,j=0}^{2l-1}\right), &k=2l-1.
\end{array}
\right.
\end{equation*}
satisfying the ``partial-skew-orthogonality'' relation:
 \begin{equation*}
\begin{aligned}
&\iint P^B_{2l}(x)y^j\frac{y-x}{x+y}d\mu(x)d\mu(y)=\frac{\tau_{2l+2}}{\tau_{2l}}\delta_{2l+1,j},\qquad 0\leq j\leq 2l+1,\\
& \iint P^B_{2l-1}(x)y^j\frac{y-x}{x+y}d\mu(x)d\mu(y)=-\frac{\tau_{2l}}{\tau_{2l-1}}\beta_j,\qquad 0\leq j\leq 2l-1,
\end{aligned}
\end{equation*}
 and $P_k^{C}(z)$ is the $k$-th monic nonsymmetric CBOP with the measures $d\mu_1(x)=d\mu(x)$ and $d\mu_2(x)=xd\mu(x)$, i.e.
\begin{align}\label{exp:cbop_p}
P_k^{C}(z)=\frac{1}{F_k^{(0,1)}}\left|
\begin{array}{ccccc}
I_{01}&I_{02}&\cdots&I_{0,k-1}&1\\
I_{11}&I_{12}&\cdots&I_{1,k-1}&z\\
\vdots&\vdots&\ddots&\vdots&\vdots\\
I_{k1}&I_{k2}&\cdots&I_{k,k-1}&z^k\\
\end{array}
\right|,
\end{align}
satisfying the biorthogonallity relation $\iint\frac{P_k^C(x)Q_j^C(y)}{x+y}yd\mu(x)d\mu(y)=\frac{F_{k+1}^{(1,0)}}{F_k^{(1,0)}}\delta_{k,j}$
in conjunction with
\begin{align*}
Q_k^{C}(z)=\frac{1}{F_k^{(0,1)}}\left|
\begin{array}{cccc}
I_{01}&I_{02}&\cdots&I_{0k}\\
I_{11}&I_{12}&\cdots&I_{1k}\\
\vdots&\vdots&\ddots&\vdots\\
I_{k-1,1}&I_{k-1,2}&\cdots&I_{k-1,k}\\
1&z&\cdots&z^{k}
\end{array}
\right|.
\end{align*}

Recall that it was shown in \cite{chang2018partial,chang2021two} that one-parameter deformations on the measures for nonsymmetric CBOPs and PSOPs induce integrable lattices. 
And, as mentioned in the above subsection, the Novikov peakon lattice describes an isospectral deformation associated with the discrete dual cubic string. The observation above suggests a novel way to study some properties of the nonsymmetric CBOPs and PSOPs together with their corresponding integrable lattices with the help of the transition matrices $a^{(k)}$.

\subsubsection{Relations with nonsymmetric CBOPs $\&$ PSOPs together with their recurrence relations}
We start from the transition between two consecutive  matrices $a^{(k-1)}(z)$ and $a^{(k)}(z)$ in \eqref{exp:mat_a}, i.e.
\begin{align*}
&a^{(k)}(z)=a^{(k-1)}(z)G_{N-k+1}L_{N-k}(z),
\end{align*}
from which we can get
\begin{align*}
\begin{pmatrix}
a_{31}^{(k)},a_{32}^{(k)},a_{33}^{(k)}
\end{pmatrix}=
\begin{pmatrix}
a_{31}^{(k-1)},a_{32}^{(k-1)},a_{33}^{(k-1)}
\end{pmatrix}
\begin{pmatrix}
1-\frac{1}{2}zl_{N-k}g_{N-k+1}^2&g_{N-k+1}&\frac{1}{2}g_{N-k+1}^2\\
-\frac{1}{2}zl_{N-k}g_{N-k+1}&1&g_{N-k+1}\\
-zl_{N-k}&0&1
\end{pmatrix}.
\end{align*}
The second column of the above equality yields 
\begin{equation}\label{rel:a_31_32}
a_{31}^{(k-1)}=\frac{1}{g_{N-k+1}}\left(a_{32}^{(k)}-a_{32}^{(k-1)}\right).
\end{equation}
Subsequently, looking at the first column with the help of the third column, one can get
\begin{equation*}
\begin{aligned}
a_{33}^{(k)}&=\frac{-1}{zl_{N-k}g_{N-k}}a_{32}^{(k+1)}+\left(\frac{1}{zl_{N-k}}\left(\frac{1}{g_{N-k}}+\frac{1}{g_{N-k+1}}\right)\right)a_{32}^{(k)}-\frac{1}{zl_{N-k}g_{N-k+1}}a_{32}^{(k-1)}.
\end{aligned}
\end{equation*}
Finally, by checking the third columns of both sides after replacing $a_{31}^{(k)}$ and $a_{33}^{(k)}$ by $a_{32}^{(k)}$, it is not hard to see that the $(3,2)$ entries  $a_{32}^{(k)}(z)$ satisfy a four-term recurrence relation
\begin{equation}\label{rec:a_32}
\begin{aligned}
z\left(a_{32}^{(k)}+a_{32}^{(k-1)}\right)=&\frac{-2}{l_{N-k}g_{N-k}g_{N-k+1}}a_{32}^{(k+1)}\\
+&\frac{2}{g_{N-k+1}}\left(\frac{1}{l_{N-k}}\left(\frac{1}{g_{N-k}}+\frac{1}{g_{N-k+1}}\right)+\frac{1}{l_{N-k+1}g_{N-k+1}}\right)a_{32}^{(k)}\\
-&\frac{2}{g_{N-k+1}}\left(\frac{1}{l_{N-k+1}}\left(\frac{1}{g_{N-k+1}}+\frac{1}{g_{N-k+2}}\right)+\frac{1}{l_{N-k}g_{N-k+1}}\right)a_{32}^{(k-1)}\\
+&\frac{2}{l_{N-k+1}g_{N-k+1}g_{N-k+2}}a_{32}^{(k-2)}.
\end{aligned}
\end{equation}
In summary, we obtain a correspondence between the discrete dual cubic string and PSOPs.
\begin{theorem} 
The $(3,2)$ entries  $a_{32}^{(k)}(z)$ can be expressed in terms of the monic PSOPs $\{P_k^B(z)\}_{k=0}^{N-1}$ defined in \eqref{psop} together with the four-term recurrence relation in terms of $\{v_k,d_k\}$
\begin{equation}\label{rec:psop}
z(P_k^B-v_kP_{k-1}^B)=P_{k+1}^B+(d_{k}-v_k)P_k^B+v_k(d_k-v_{k+1})P_{k-1}^B+(v_k)^2v_{k-1}P_{k-2}^B
\end{equation}
with initial values 
$$P_{-2}^B=P_{-1}^B=0,\quad P_0^B=1,$$
according to
\begin{subequations}
\begin{align}
&a_{32}^{(k)}(z)=(-1)^kzl_Ng_N\prod_{j=1}^{k-1}\frac{l_{N-j}g_{N-j}g_{N-j+1}}{2}P_{k-1}^B(z),\nonumber\\
&v_k=\frac{2}{l_{N-k}g_{N-k}g_{N-k+1}},\label{trans:v_gl}\\
&d_k=\frac{2}{g_{N-k}}\left(\frac{1}{l_{N-k}}\left(\frac{1}{g_{N-k}}+\frac{1}{g_{N-k+1}}\right)+\frac{1}{l_{N-k-1}}\left(\frac{1}{g_{N-k}}+\frac{1}{g_{N-k-1}}\right)\right),\label{trans:d_gl}
\end{align}
\end{subequations}
with the convention $g_0=g_{N+1}=+\infty.$
\end{theorem}

By virtue of \eqref{rel:a_31_32}, the fromula \eqref{rec:a_32} can be rewritten as 
\begin{equation*}
\begin{aligned}
&z\left(a_{32}^{(k)}+a_{32}^{(k-1)}\right)=\frac{-2}{l_{N-k}g_{N-k+1}}a_{31}^{(k)}+\frac{2}{g_{N-k+1}}\left(\frac{1}{l_{N-k}}+\frac{1}{l_{N-k+1}}\right)a_{31}^{(k-1)}-\frac{2}{l_{N-k+1}g_{N-k+1}}a_{31}^{(k-2)},
\end{aligned}
\end{equation*}
from which, subtracting the formula with $k$ replaced by $k-1$, we are led to a four-term recurrence for $a_{31}^{(k)}$
\begin{equation*}
\begin{aligned}
&z\left(g_{N-k+1}a_{31}^{(k-1)}+g_{N-k+2}a_{31}^{(k-2)}\right)\\
=&\frac{-2}{l_{N-k}g_{N-k+1}}a_{31}^{(k)}+\left(\frac{2}{g_{N-k+1}}\left(\frac{1}{l_{N-k}}+\frac{1}{l_{N-k+1}}\right)+\frac{2}{l_{N-k+1}g_{N-k+2}}\right)a_{31}^{(k-1)}\\
&-\left(\frac{2}{l_{N-k+1}g_{N-k+1}}+\frac{2}{g_{N-k+2}}\left(\frac{1}{l_{N-k+1}}+\frac{1}{l_{N-k+2}}\right)\right)a_{31}^{(k-2)}+\frac{2}{l_{N-k+2}g_{N-k+2}}a_{31}^{(k-3)}.
\end{aligned}
\end{equation*}
Therefore, we arrive at the following theorem.
\begin{theorem}The $(3,1)$ entries  $a_{31}^{(k)}(z)$ can be expressed in terms of the monic nonsymmetric CBOPs $\{P_k^C(z)\}_{k=0}^{N}$ given in \eqref{exp:cbop_p} together with the four-term recurrence relation in terms of $\{v_k,d_k\}$
\begin{equation}\label{rec:cbop01}
z(P_k^C-v_kP_{k-1}^C)=P_{k+1}^C+(d_{k}-v_{k+1})P_k^C+v_k(d_{k-1}-v_{k-1})P_{k-1}^C+(v_{k-1})^2v_{k}P_{k-2}^C,
\end{equation}
with initial values 
$$P_{-2}^C=P_{-1}^C=0,\quad P_0^C=1,$$
according to
\begin{align*}
&a_{31}^{(k)}(z)=(-1)^{k-1}\frac{zl_Ng_N}{g_{N-k}}\prod_{j=1}^{k}\frac{l_{N-j}g_{N-j}g_{N-j+1}}{2}P_{k}^C(z),\\
&v_k=\frac{2}{l_{N-k}g_{N-k}g_{N-k+1}},\\
&d_k=\frac{2}{g_{N-k}}\left(\frac{1}{l_{N-k}}\left(\frac{1}{g_{N-k}}+\frac{1}{g_{N-k+1}}\right)+\frac{1}{l_{N-k-1}}\left(\frac{1}{g_{N-k}}+\frac{1}{g_{N-k-1}}\right)\right),
\end{align*}
with the convention $g_0=g_{N+1}=+\infty.$
\end{theorem}

So far, we have established the connections between the discrete dual cubic string and PSOPs together with their recurrence relations, and the corresponding relationships between the discrete dual cubic string and a sequence of specific nonsymmetric CBOPs as well. Again, it is noted that we haven't yet considered the time dependence in this section. The remaining part of this section will involve time flows.

\subsubsection{Novikov peakon and finite B-Toda: an alternative view} The Novikov peakon lattice \eqref{NV_eq:peakon} describes an isospectral flow under the deformation of the measure $d\mu(x,t)=\sum_{j=1}^Nb_j(0)e^{\frac{t}{x}}\delta(x-\zeta_j)dx$ associated with the discrete dual cubic string. As shown in \cite{chang2018partial,chang2018application}, the PSOPs can induce the finite B-Toda lattice 
\begin{equation}
\label{eq:finite_btoda}
\begin{aligned}
&\dot v_{k}=v_{k}(d_{k}-d_{k-1}),  \qquad\qquad\qquad\qquad\ \  \ k=1,\dots,N-1,\\
&\dot d_k=v_{k+1}(d_{k+1}+d_k)-v_{k}(d_{k}+d_{k-1}), \quad k=0,\dots,N-1,
\end{aligned}
\end{equation}
by deforming the measure $d\mu(x,t)=\sum_{j=1}^Nb_j(0)e^{tx}\delta(x-\zeta_j)dx$. 

Combing \eqref{form:xm},  \eqref{map:xm_yg}, \eqref{trans:v_gl} and \eqref{trans:d_gl}, in conjunction with time dependence, we obtain the following theorem concerning the Novikov peakon and finite B-Toda lattices. 
\begin{theorem}
Let $\tau_k^{(l)}(0)$ be Pfaffians defined in \eqref{def:tau} with the discrete measure 
$$d\mu(x,0)=\sum_{j=1}^Nb_j(0)\delta(x-\zeta_j(0))dx.$$
More exactly,
$$\tau_{k}^{(l)}(0)=
\left\{
\begin{array}{ll}
\Pf\left(\left(J_{l+i,l+j}(0)\right)_{i,j=0}^{2m-1}\right),&k=2m,\\
\ \\
\Pf\left(\left(
\begin{array}{ccc}
0&\vline&
\begin{array}{c}
\beta_{l+j}(0)
\end{array}\\
\hline
-\beta_{l+i}(0)
&\vline&J_{l+i,l+j}(0)
\end{array}
\right)_{i,j=0}^{2m-1}\right),&k=2m-1.
\end{array}
\right.
$$
where
\begin{align*}
J_{i,j}(0)=\iint \frac{y-x}{x+y}x^iy^jd\mu(x,0)d\mu(y,0),\qquad \beta_j(0)=\int x^id\mu(x,0)
\end{align*}
with
 \[
 0<\zeta_1(0)<\zeta_2(0)<\cdots<\zeta_N(0),\qquad b_p(0)>0.
 \]
Define
$$W_k^{(l)}(0)=\tau_k^{(l+1)}(0)\tau_k^{(l)}(0)-\tau_{k-1}^{(l+1)}(0)\tau_{k+1}^{(l)}(0).$$
 \begin{enumerate}[(i)]
 \item \label{NVB_i1} Introduce the variables $\{x_k(0),m_k(0)\}_{k=1}^N$ defined by
$$ x_{k'}(0)=\frac{1}{2}\ln\frac{W_k^{(-1)}(0)}{W_{k-1}^{(0)}(0)},\qquad m_{k'}(0)=\frac{\sqrt{W_k^{(-1)}(0)W_{k-1}^{(0)}(0)}}{\tau_k^{(0)}(0)\tau_{k-1}^{(0)}(0)},$$
where $k'=n+1-k.$ If $\{\zeta_p(t),b_p(t)\}_{p=1}^N$ evolve as 
 \[
 \dot \zeta_p=0,\qquad \dot b_p=\frac{b_p}{\zeta_p},
 \]
then $\{x_k(t),m_k(t)\}_{k=1}^N$  satisfy the Novikov peakon ODEs \eqref{NV_eq:peakon}. 
\item \label{NVB_i2}
 Introduce the variables $\{\{v_k(0)\}_{k=1}^{N-1},\{d_{k}(0)\}_{k=0}^{N-1}\}$ defined by
\begin{align*}
v_k(0)=&\frac{\tau_{k+1}^{(0)}(0)\tau_{k-1}^{(0)}(0)}{({\tau_{k}^{(0)}(0)})^2},\\
 d_k(0)=&\frac{\tau_{k+1}^{(0)}(0)\tau_{k-1}^{(0)}(0)}{({\tau_{k}^{(0)}(0)})^2}+\frac{\left(\tau_{k+1}^{(0)}(0)\right)^2\left(W_{k}^{(-1)}(0)+W_{k-1}^{(0)}(0)\right)}{\left(\tau_{k}^{(0)}(0)\right)^2\left(W_{k+1}^{(-1)}(0)+W_{k}^{(0)}(0)\right)}\\
&+\frac{\tau_{k+2}^{(0)}(0)\tau_{k}^{(0)}(0)}{({\tau_{k+1}^{(0)}(0)})^2}+\frac{\left(\tau_{k}^{(0)}(0)\right)^2\left(W_{k+2}^{(-1)}(0)+W_{k+1}^{(0)}(0)\right)}{\left(\tau_{k+1}^{(0)}(0)\right)^2\left(W_{k+1}^{(-1)}(0)+W_{k}^{(0)}(0)\right)}.
\end{align*}
 If $\{\zeta_p(t),b_p(t)\}_{p=1}^N$ evolve as 
 \[
 \dot \zeta_p=0,\qquad \dot b_p=\zeta_pb_p,
 \]
 then  $\{\{v_k(0)\}_{k=1}^{N-1},\{d_{k}(0)\}_{k=0}^{N-1}\}$  satisfy the finite B-Toda lattice \eqref{eq:finite_btoda} with $v_0=v_N=0$.
 \item  \label{NVB_i3} There exists a mapping from $\{x_k(0),m_k(0)\}_{k=1}^N$ to  $\left\{\{v_k(0)\}_{k=1}^{N-1},\{d_{k}(0)\}_{k=0}^{N-1}\right\}$ according to
 \begin{align*}
 &v_k(0)=\frac{1}{2m_{k'}(0)m_{k'-1}(0)\cosh x_{k'}(0)\cosh x_{k'-1}(0)(\tanh x_{k'}(0)-\tanh x_{k'-1}(0))}, \\
 &d_k(0)=v_k(0)\left(1+\frac{m_{k'}(0)\cosh x_{k'}(0)}{m_{k'-1}(0)\cosh x_{k'-1}(0)}\right)+v_{k+1}(0)\left(1+\frac{m_{k'-2}(0)\cosh x_{k'-2}(0)}{m_{k'-1}(0)\cosh x_{k'-1}(0)}\right)
\end{align*}
with the convention
$$ x_0(0)=-\infty,\quad x_{N+1}(0)=+\infty.$$
 \end{enumerate}
 \end{theorem}
\begin{proof} Item \eqref{NVB_i1} immediately follows from \eqref{form:xm} and the time dependence \eqref{evo:zb}. Item \eqref{NVB_i3} on the correspondence between the variables $\{x_k,m_k\}$ and $\{v_k,d_k\}$ can be derived from \eqref{map:xm_yg}, \eqref{trans:v_gl} and \eqref{trans:d_gl}. The conclusion \eqref{NVB_i2} for the finite B-Toda lattice is a consequence of combing \eqref{NVB_i1} and  \eqref{NVB_i3} with its own time dependence.
\end{proof}
\begin{remark}
By using certain identities, one can show that the formulae in the above theorem are equivalent to those in \cite[Theorem B.1]{chang2018degasperis}.  This presents an alternative way to connect the Novikov peakon and finite B-Toda lattices as opposite flows.
\end{remark}

\subsubsection{A connection between two integrable lattices}
The recurrence relation for nonsymmetric CBOPs together with one isospectral deformation was presented in \cite[Proposition 3.4]{chang2021two}, from which
we obtain that two sequences of monic nonsymmetric CBOPs $\{P_k^C(z)\}$ and $\{Q_k^C(z)\}$ with $d\mu_2(x)=xd\mu_1(x)=xd\mu(x)$, satisfy the recurrence relations
\begin{align}
z(P_{k}^C(z)-\mathcal{A}_{k-1}P_{k-1}^C(z))&=P_{k+1}^C(z)+\mathcal{B}_{k-1}P_{k}^C(z)+\mathcal{C}_{k-1}P_{k-1}^C(z)+\mathcal{D}_{k-1}P_{k-2}^C(z), \label{rec:cbop01_ABCD1}\\
z(Q_{k}^C(z)-\mathcal{\hat A}_{k-1}Q_{k-1}^C(z))&=Q_{k+1}^C(z)+\mathcal{\hat B}_{k-1}Q_{k}^C(z)+\mathcal{\hat C}_{k-1}Q_{k-1}^C(z)+\mathcal{\hat D}_{k-1}Q_{k-2}^C(z),\nonumber
\end{align}
where the coefficients are given by
\begin{align*}
\begin{aligned}
&\mathcal{A}_{k}=\hat C_k,\quad \mathcal{B}_{k}=-\hat C_k+B_{k+1},\quad \mathcal{C}_{k}=-A_{k+1}+\hat C_n\hat B_k,\quad \mathcal{D}_{k}=A_k\hat C_k,\\
&\mathcal{\hat A}_{k}= C_k,\quad\mathcal{\hat B}_{k}=-{C}_k+\hat B_{k+1},\quad\mathcal{\hat C}_{k}=-A_{k+1}+C_n B_n,\quad \mathcal{\hat D}_{k}=A_k C_k,
\end{aligned}
\end{align*}
with 
\begin{align*}
&A_k=\frac{F_{k+1}^{(0,1)}F_{k-1}^{(0,1)}}{\left(F_{k}^{(0,1)}\right)^2},\quad B_k=\frac{E_{k+1}^{(1,0)}}{F_{k+1}^{(0,1)}}-\frac{E_{k}^{(1,0)}}{F_{k}^{(0,1)}},\quad \hat B_n=\frac{E_{k+1}^{(0,1)}}{F_{k+1}^{(0,1)}}-\frac{E_{k}^{(0,1)}}{F_{k}^{(0,1)}},\\
& C_k=\frac{G_{k+2}^{(1,0)}F_{k}^{(0,1)}}{G_{k+1}^{(1,0)}F_{k+1}^{(0,1)}},\quad \hat C_k=\frac{G_{k+2}^{(0,1)}F_{k}^{(0,1)}}{G_{k+1}^{(0,1)}F_{k+1}^{(0,1)}}.
\end{align*}
Here $F_k^{(i,j)},G_k^{(i,j)}$ denote the determinants in Definition \ref{def:FG} and each $E_k^{(i,j)}$ denotes a determinant of a matrix obtained by deleting the last row and the $k$-th column from $F_{k+1}^{(i,j)}$.

Recall that we have derived a recurrence relation for $\{P_k^C(z)\}$ with the coefficients consisting of $\{v_k,d_k\}$, which differs from the expressions by using $\{A_k,B_k,\hat B_k,C_k,\hat C_k\}$ in \eqref{rec:cbop01_ABCD1}. Therefore, comparing the formulae in \eqref{rec:cbop01} and \eqref{rec:cbop01_ABCD1} will result in a correspondence between these two groups of variables.

In fact, by comparing \eqref{rec:cbop01} and \eqref{rec:cbop01_ABCD1}, it easy to see that the following relations are satisfied:
\begin{equation}\label{trans:ABC_ub}
A_{k}=v_{k}^2,\quad  B_{k}=d_k+v_k-v_{k+1}, \quad \hat B_{k}=d_k+v_{k+1}-v_k, \quad \hat C_k=v_{k+1}.
\end{equation}
In order to establish a bridge between $\{C_k\}$ and $\{v_k,d_k\}$, we introduce the time deformation of the measure $d\mu(x;t)=e^{xt}d\mu(x)$. 

It is known in \cite{chang2018partial} that PSOPs will evolve according to
\begin{equation}\label{evo:psop}
\dot P_k^B(z;t)-v_k\dot P_{k-1}^B(z;t)=-v_k(d_k+d_{k-1})P_{k-1}^B(z;t).
\end{equation}
The compatibility of \eqref{rec:psop} and \eqref{evo:psop} leads to the B-Toda lattice
\begin{align}\label{eq:btoda}
&\dot v_{k}=v_{k}(d_{k}-d_{k-1}),  \qquad \dot d_k=v_{k+1}(d_{k+1}+d_k)-v_{k}(d_{k}+d_{k-1}).
\end{align}
On the other hand, the compatibility between the four-term recurrence relation \eqref{rec:cbop01_ABCD1} for the general nonsymmetric CBOPs and the time evolution
\begin{equation}\label{evo:CBOP}
\begin{aligned}
\dot
P_{k}^C(z;t)-\hat C_{k-1}\dot P_{k-1}^C(z;t)&=-\hat C_{k-1}(B_{k-1}+\hat B_{k-1})P_{k-1}^C(z;t),\\ 
\dot Q_{k}^C(z;t)-{C}_{k-1}\dot Q_{k-1}^C(z;t)&= -C_{k-1}(B_{k-1}+\hat B_{k-1})Q_{k-1}^C(z;t),
\end{aligned}
\end{equation}
 yields the following integrable lattice \cite[(eq. 3.13)]{chang2021two}
\begin{equation}\label{eq:is_ABC}
\begin{aligned}
&\dot A_k=A_k(B_k+\hat B_k-B_{k-1}-\hat B_{k-1}),\\
&\dot B_k=(B_{k}+\hat B_{k})\hat C_k-(B_{k-1}+\hat B_{k-1})\hat C_{k-1},\\
&\dot{\hat{B}}_k=(B_{k}+\hat B_{k}) C_k-(B_{k-1}+\hat B_{k-1})C_{k-1},\\
&\dot C_k=C_k\left(C_{k+1}-\frac{A_{k+1}}{C_k}+\hat B_{k+1}-C_k+\frac{A_k}{C_{k-1}}-\hat B_{k}\right),\\
&\dot{\hat{C}}_k=\hat C_k\left(\hat C_{k+1}-\frac{A_{k+1}}{\hat C_k}+ B_{k+1}-\hat C_k+\frac{A_k}{\hat C_{k-1}}- B_{k}\right).
\end{aligned}
\end{equation}
With the help of the observed relation \eqref{trans:ABC_ub}, we obtain a correspondence between two integrable lattices, i.e. \eqref{eq:is_ABC} and \eqref{eq:btoda}.
\begin{theorem}
The integrable lattice \eqref{eq:is_ABC} is reduced to the B-Toda lattice \eqref{eq:btoda} under the constraints
\begin{equation}\label{rel:abc_vd}
A_{k}=v_{k}^2,\ \  B_{k}=d_k+v_k-v_{k+1}, \ \ \hat B_{k}=d_k+v_{k+1}-v_k, \ \ C_k=\frac{v_{k+1}d_{k+1}}{d_k}, \ \ \hat C_k=v_{k+1}.
\end{equation}
\end{theorem}
\begin{proof}
The proof is almost straightforward except the derivation of the expression for $C_k$ in terms of $v_k,d_k$, which can be deduced from the third equation in \eqref{eq:is_ABC}.
\end{proof}
We end this section by proposing an alternative Lax pair of the B-Toda lattice (rather than \eqref{rec:psop} and \eqref{evo:psop} originally obtained in \cite{chang2018partial}), which is suggested by the above correspondence \eqref{rel:abc_vd} in conjunction with the compatible system  consisting of \eqref{rec:cbop01_ABCD1} and \eqref{evo:CBOP}. More exactly, the B-Toda lattice can be obtained from the compatibility of the following overdetermined system:
\begin{align*}
zL_1\Phi(z;t)=L_2\Phi(z;t),\qquad L_1\dot \Phi(z;t) =R_2\Phi(z;t),
\end{align*}
where 
\begin{eqnarray*}
&\Phi(z;t)=(P_0^C(z;t),P_1^C(z;t),\ldots)^\mathrm{T}\\
&L=L_1^{-1}L_2,\qquad R=L_1^{-1}R_2,\\
&L_1=\left(\begin{array}{cccccc}
1&&\\
-v_1&1&\\
&-v_2&1\\
&&\ddots&\ddots
\end{array}
\right),\quad
R_2=\left(\begin{array}{cccccc}
0&&&\\
\mathcal{A}_0&0&&\\
&\mathcal{A}_1&0&&\\
&&\ddots&\ddots
\end{array}
\right),\\
&L_2=\left(\begin{array}{cccccc}
\mathcal{B}_{-1}&1&&&\\
\mathcal{C}_0&\mathcal{B}_0&1&&\\
\mathcal{D}_1&\mathcal{C}_1&\mathcal{B}_1&1\\
&\ddots&\ddots&\ddots&\ddots
\end{array}
\right)
\end{eqnarray*}
with $$\mathcal{A}_{k}=-2d_kv_{k+1},\quad \mathcal{B}_{k}=d_{k+1}-v_{k+2},\quad \mathcal{C}_{k}=v_{k+1}(d_k-v_k),\quad \mathcal{D}_{k}=v_k^2v_{k+1},$$
in other words, the B-Toda lattice admits the Lax representation:
$$\dot L=[R,L],$$
where 
$$L=L_1^{-1}L_2,\qquad R=L_1^{-1}R_2.$$
\begin{remark}
The results obtained in this subsection are not restricted to the finite case. The finite case will be covered by truncation.
\end{remark}
\section*{Acknowledgements}
We are grateful for valuable comments from Jacek Szmigielski and Xing-Biao Hu. We thank Sergio Medina Peralta for short conversations in the embryonic stage of this project. We also would like to express our deep gratitude to the anonymous referees for many constructive remarks and suggestions, which substantially improved the quality of this paper. This work was supported in part by the National Natural Science Foundation of China (\#12171461, 11688101, 11731014) and the Youth Innovation Promotion Association CAS. 

\begin{appendix}
\section{On the Pfaffian}\label{app_pf} 
 \textit{Pfaffian} is a more general algebraic tool than \textit{determinant} although not as well known \cite{Caieniello1973,hirota2004direct}. For the sake of convenience, we include some useful materials on Pfaffians in this appendix. 
\begin{define}\label{def_pf}
Given a skew-symmetric matrix of order $2N$: $A=(a_{ij})_{i,j=1}^{2N}$,  a Pfaffian of order $N$ is defined according to the formula
\begin{align*}
\Pf(A)=\sum_{P}(-1)^Pa_{i_1,i_2}a_{i_3,i_4}\cdots a_{i_{2N-1},i_{2N}}.
\end{align*}
The summation means the sum over all possible combinations of pairs selected from $\{1,2,\ldots,2N\}$ satisfying
\begin{align*}
&i_{2l-1}<i_{2l+1},\qquad i_{2l-1}<i_{2l}.
\end{align*}
The factor $(-1)^P$ takes the value $+1 (-1)$ if the sequence $i_1, i_2, . . . , i_{2N}$ is an even (odd) permutation of $1, 2, . . . , 2N.$ Usually, there are the conventions that the Pfaffian of order $0$ is 1 and that for negative order is $0$.
\end{define}

\begin{remark}
For an antisymmetric matrix $A$, the square of its Pfaffian is its determinant, i.e.  
\begin{align*} 
(\Pf(A))^2=\det(A),
\end{align*} 
which was first proved by Cayley \cite{cayley1849}. For example,
\begin{enumerate}[(i)]
\item For $N=1$, we have $\Pf(A)=a_{12}$, and $\det(A)=(a_{12})^2$.
\item For $N=2$, we have $\Pf(A)=a_{12}a_{34}-a_{13}a_{24}+a_{14}a_{23}$ and $\det(A)=(a_{12}a_{34}-a_{13}a_{24}+a_{14}a_{23})^2$.
\end{enumerate}
\end{remark} 

\begin{remark}\label{rem:pf1}
There exist many notations for the Pfaffian $\Pf(A)$. For example, there is a notation due to Hirota \cite{hirota2004direct}, that is,
$$\Pf(1,2,\ldots,2N)$$
based on the Pfaffian entries $\Pf(i,j)$, denotes the Pfaffian $\Pf(A)$ for the skew-symmetric matrix $A=(a_{ij})_{i,j=1}^{2N}$ with $a_{ij}=\Pf(i,j)$. This notation has been used in many references, e.g. \cite{adler1999pfaff} and our previous works \cite{chang2018partial,chang2018application}, and it turns out that it can bring more convenience in some occasions. However, we would use the well-known notation $\Pf(A)$ since it is convenient enough here. 
\end{remark}  

A Pfaffian has some basic properties similar to those for a determinant. For example, a Pfaffian admits similar Laplace expansion. 
\begin{prop}
For any antisymmetric matrix $A=(a_{ij})_{i,j=1}^{2N}$ and a fixed integer $i$, $1\leq i\leq 2N$, we have
\begin{eqnarray*}
\Pf(A)&=&\sum_{1\leq j\leq 2N,j\neq i}(-1)^{i+j-1}a_{i,j}\Pf(A_{\hat{i},\hat{j}}),
\end{eqnarray*}
where $A_{\hat{i},\hat{j}}$ denotes the antisymmetric matrix of order $2N-2$ obtained by deleting the $i,j$-th rows and columns from $A$. 
If we choose $i$ as $1$ or $2N$, we have
\begin{eqnarray*}
\Pf(A)
&=&\sum_{j=2}^{2N}(-1)^ja_{1,j}\Pf(A_{\hat{1},\hat{j}})=\sum_{j=1}^{2N-1}(-1)^{j+1}a_{j,2N}\Pf(A_{\hat{j},\widehat{2N}}).
\end{eqnarray*}
\end{prop}

It is not hard to observe the following facts, which are similar to those on row/column operations for an determinant.
\begin{prop}
\begin{enumerate}
\item Multiplication of a row and a column by a constant is equivalent to multiplication of the Pfaffian by the same constant.
\item Simultaneous interchange of the two different rows and corresponding columns changes the sign of the Pfaffian.
\item A multiple of a row and corresponding column added to another row and corresponding column does not change the value of the Pfaffian.
\end{enumerate} 
\end{prop}
Furthermore, a Pfaffian has the following property, which can be obtained based on the above properties for row/column operations. 
\begin{prop}
For any antisymmetric matrix $A$ of order $2N$, one has the equality $$\Pf(B AB^T)=\det(B)\Pf(A)$$ for any matrix $B$ of order $2N$.
\end{prop}

\subsection{Formulae related to determinants and Pfaffians}
There are many interesting connections between some specific determinants and Pfaffians. The following identities play important roles in our paper.
\begin{enumerate}[I.]
\item
 For a skew-symmetric matrix $A_{2N-1}=(a_{ij})_{i,j=1}^{2N-1}$ augmented with an arbitrary row and column, the following identity holds 
 \begin{align}\label{det_pf_odd}
\det\left(\left(
\begin{array}{ccc}
a_{i,j}&\vline&
\begin{array}{c}
x_i
\end{array}\\
\hline
-y_j
&\vline&z
\end{array}
\right)_{i,j=1}^{2N-1}\right)=\Pf(B_{2N})\Pf(C_{2N}),
 \end{align}
 where 
 $$
 B_{2N}=\left(
\begin{array}{ccc}
a_{i,j}&\vline&
\begin{array}{c}
x_i
\end{array}\\
\hline
-x_j
&\vline&0
\end{array}
\right)_{i,j=1}^{2N-1},\quad  C_{2N}=\left(
\begin{array}{ccc}
a_{i,j}&\vline&
\begin{array}{c}
y_i
\end{array}\\
\hline
-y_j
&\vline&0
\end{array}
\right)_{i,j=1}^{2N-1}.
$$

\item For a skew-symmetric matrix $A_{2N}=(a_{ij})_{i,j=1}^{2N}$ augmented with an arbitrary row and column, the following identity holds
 \begin{align}\label{det_pf_even}
\det\left(\left(
\begin{array}{ccc}
a_{i,j}&\vline&
\begin{array}{c}
x_i
\end{array}\\
\hline
y_j
&\vline&z
\end{array}
\right)_{i,j=1}^{2N}
\right)=\Pf(A_{2N})\Pf(B_{2N+2}),
 \end{align}
 where 
 $$
 B_{2N+2}=\left(
\begin{array}{cccc}
a_{i,j}&\vline&
\begin{array}{cc}
x_i&y_i
\end{array}\\
\hline
\begin{array}{ccc}
-x_j&\\
-y_j&
\end{array}
&\vline&
\begin{array}{cc}
0&z\\
-z&0
\end{array}
\end{array}
\right)_{i,j=1}^{2N}.
$$
\end{enumerate}
\begin{remark}
In fact, these two formulae can be obtained by applying the well known Jacobi determinant  identity  \cite{aitken1959determinants,hirota2004direct} to the determinants constructed from certain skew-symmetric matrices. The Jacobi determinant  identity reads
\begin{align*}
\mathcal{D} \mathcal{D}\left(\begin{array}{cc}
i_1 & i_2 \\
j_1 & j_2 \end{array}\right)=\mathcal{D}\left(\begin{array}{c}
i_1  \\
j_1 \end{array}\right)\mathcal{D}\left(\begin{array}{c}
i_2  \\
j_2 \end{array}\right)-\mathcal{D}\left(\begin{array}{c}
i_1  \\
j_2 \end{array}\right)\mathcal{D}\left(\begin{array}{c}
i_2  \\
j_1 \end{array}\right),
\end{align*}
where $\mathcal{D}$ is an indeterminate determinant and $\mathcal{D}\left(\begin{array}{cccc}
i_1&i_2 &\cdots& i_k\\
j_1&j_2 &\cdots& j_k
\end{array}\right)$ with $ i_1<i_2<\cdots<i_k,\ j_1<j_2<\cdots<j_k$ denotes the
determinant of the matrix obtained from $\mathcal{D}$ by removing the rows with indices
$i_1,i_2,\dots, i_k$ and the columns with indices $j_1,j_2,\dots, j_k$.
\end{remark}

\subsection{Formulae for Schur's Pfaffians}
Schur's Pfaffian identity \cite{ishikawa2006generalizations,knuth1996over,okada2019pfaffian,schur1911uber} reads
\begin{align}
\Pf\left(\left(\frac{s_i-s_j}{s_i+s_j}\right)_{i,j=1}^{2N}\right)=\prod_{1\leq i<j\leq 2N}\frac{s_i-s_j}{s_i+s_j},\label{id_schur}
\end{align}
for even case (see \cite[Section 4]{knuth1996over} for a proof). For the odd case, one has 
\begin{align}
\Pf\left(\left(
\begin{array}{ccc}
0&\vline&
\begin{array}{c}
1
\end{array}\\
\hline
-1
&\vline&\frac{s_i-s_j}{s_i+s_j}
\end{array}
\right)_{i,j=1}^{2N-1}\right)=\prod_{1\leq i<j\leq 2N-1}\frac{s_i-s_j}{s_i+s_j},\label{id_schur_odd}
\end{align}
which can be proved by setting $s_{2N}=0$ (or $s_{1}=0$) in the even case. These identities play central roles in enumerative combinatorics and representation theory.
\end{appendix}

\small
\bibliographystyle{abbrv}
\bibliographystyle{plain}
\def\cydot{\leavevmode\raise.4ex\hbox{.}}
  \def\cydot{\leavevmode\raise.4ex\hbox{.}} \def\cprime{$'$}

\end{document}